\documentclass{article}
\usepackage{fullpage}

\usepackage{amsmath, amssymb,bm}
\usepackage{hyperref}
\usepackage{enumerate,mathtools}
\usepackage{amsthm}
\usepackage{cite}
\usepackage{microtype}
\usepackage{color}
\usepackage{subcaption}

\newtheorem{theorem}{Theorem}
\newtheorem{cor}[theorem]{Corollary}
\newtheorem{lemma}[theorem]{Lemma}

\newtheorem{prop}{Proposition}

\theoremstyle{definition}

\newtheorem{remark}{Remark}
\newtheorem{assumption}{Assumption}

\usepackage{pgfplots}
\usetikzlibrary{positioning}
\pgfplotsset{compat=1.12}

\newcommand{\bA}{\bm{A}}
\newcommand{\bB}{\bm{B}}

\newcommand{\bG}{\bm{G}}

\newcommand{\bN}{\bm{N}}

\newcommand{\bW}{\bm{W}}
\newcommand{\bX}{\bm{X}}
\newcommand{\bY}{\bm{Y}}
\newcommand{\bZ}{\bm{Z}}

\newcommand{\by}{\bm{y}}

\newcommand{\cA}{\mathcal{A}}

\newcommand{\cD}{\mathcal{D}}

\newcommand{\cF}{\mathcal{F}}

\newcommand{\cI}{\mathcal{I}}

\newcommand{\cL}{\mathcal{L}}

\newcommand{\cR}{\mathcal{R}}

\newcommand{\cT}{\mathcal{T}}
\newcommand{\cU}{\mathcal{U}}

\newcommand{\cW}{\mathcal{W}}

\newcommand{\dd}{\mathrm{d}}

\newcommand{\one}{\boldsymbol{1}}

\newcommand{\psd}{\mathbb{S}_+}

\newcommand{\bEx}{\ensuremath{\mathbb{E}}}

\newcommand{\ex}[1]{\ensuremath{\mathbb{E}\left[ #1\right]}}

\newcommand{\pr}[1]{\ensuremath{\mathbb{P}\left[ #1\right]}}

\DeclareMathOperator{\cov}{\mathsf{Cov}}

\DeclareMathOperator{\var}{\sf Var}

\DeclareMathOperator{\diag}{\mathrm diag}

\newcommand{\DKL}[2]{\ensuremath{D\left( #1 \, \|  \, #2 \right)}}

\newcommand{\wX}{X}

\DeclareMathOperator{\gtr}{tr}

\DeclareMathOperator{\gvec}{\mathsf{vec}}

\DeclareMathOperator{\MMSE}{\mathsf{MMSE}}

\newcommand{\reals}{\mathbb{R}}

\newcommand{\eps}{\epsilon}
\newcommand{\normal}{\mathcal{N}}

\let\originalleft\left
\let\originalright\right
\renewcommand{\left}{\mathopen{}\mathclose\bgroup\originalleft}
\renewcommand{\right}{\aftergroup\egroup\originalright}

\allowdisplaybreaks

\newif\iflongpaper

\longpapertrue

\title{The Geometry of Community Detection\\ via the MMSE Matrix} 

\author{Galen Reeves \and  Vaishakhi Mayya \and Alexander Volfovsky
\thanks{G.~Reeves is with the Department of Electrical and Computer Engineering and the Department of Statistical Science, Duke University, Durham, NC 27708 USA (e-mail: galen.reeves@duke.edu). V.~Mayya is with the Department of Electrical and Computer Engineering, Duke University, Durham, NC 27708 USA (e-mail: vaishakhi.mayya@duke.edu). A.~Volfovsky is with the Department of Statistical Science, Duke University, Durham, NC 27708 USA (e-mail: alexander.volfovsky@duke.edu).}}

\begin{document}


\maketitle

\begin{abstract}
The information-theoretic limits of community detection have been studied extensively for network models with high levels of symmetry or homogeneity. The contribution of this paper is to study a broader class of network models that allow for variability in the sizes and behaviors of the different communities, and thus better reflect the behaviors observed in real-world networks. Our results show that the ability to detect communities can be described succinctly in terms of a matrix of effective signal-to-noise ratios that provides a geometrical representation of the relationships between the different communities. This characterization follows from a matrix version of the I-MMSE relationship and generalizes the concept of an effective scalar signal-to-noise ratio introduced in previous work.  We provide explicit formulas for the asymptotic per-node mutual information and  upper bounds on the minimum mean-squared error. The theoretical results are supported by numerical simulations. 
\end{abstract}



\section{Introduction} 
\label{sec:introduction}

Modern data problems often ask questions about how individuals (or computers or countries) interact or relate to each other within a network. A frequently studied problem in this context is that of community detection: how does one partition a network into clusters (or communities or groups) of nodes? A natural partition of a network is into communities that exhibit similar connection patterns, both within and between communities. A generative model for random networks called the stochastic block model (SBM) exhibits such behavior and hence much of the theoretical analysis of community detection has focused on it \cite{holland1983stochastic}. Under the SBM each individual belongs to exactly one of $k$ communities, and the probability of an edge between two individuals is exclusively a function of their community memberships. 

The problem of community detection can be modeled in terms of a joint distribution on $(\bX, \bG)$ where $\bG$ is a simple graph on $n$ vertices and $\bX = (X_1, \dots, X_n)$ is a collection of labels associated with the vertices.
In the SBM this joint distribution is governed by two parameters: a probability vector $p$ of each node being assigned to one of $k$ labels, and a $k\times k$ matrix of probabilities $Q$ where $Q_{ab}$ is the probability of an edge between nodes in communities $a$ and $b$. The community detection task is recovering the labels $\bX$ given the graph $\bG$ and potentially side information.

Inspired by the work of Decelle et al.~\cite{decelle:2011a}, a recent line of work has studied the information-theoretic limits of recovery when the distribution of $(\bX,\bG)$ is known. Most of this work has focused on either the two-community SBM \cite{deshpande:2017,caltagirone:2018, lelarge:2018,barbier:2016a,lesieur:2017,krzakala:2016,deshpande:2018} or the so-called $k$-community symmetric SBM \cite{lesieur:2017, abbe:2018a, banks:2016a,abbe:2018}. 
In all of these cases,  performance is summarized in terms of a single numerical value, which is often referred to as the effective signal-to-noise ratio of the problem. General SBMs have been considered by Abbe and Sandon~\cite{abbe:2018a} who characterize conditions for weak recovery and also by Lesieuir et al.~\cite{lesieur:2017} who analyze the performance of an approximate message passing algorithm. 

A different line of research within the statistics community has focused on settings where the parameters of the distribution, such as the distribution of communities and the conditional probabilities of edges, are unknown quantities that must also be inferred, along with the community memberships \cite{rohe2011spectral, suwan:2016}. While the models considered in this literature are highly flexible, the conditions needed for consistent recovery of communities corresponds to a very high SNR regime relative to the information theoretic analysis. 

\subsection{Our Contributions} 

The contribution of this paper is to characterize the information-theoretic limits for a large class of degree-balanced SBMs. In contrast to the symmetric SBM, these models allow for variability in the sizes and behaviors of the different communities, and thus reflect behaviors observed in real-world networks. While previous work is limited to a scalar measure of performance for the overall community detection problem, we introduce a multivariate measure of performance, the minimum mean-squared error (MMSE) matrix, which describes detection limits for individual communities.  For example, this matrix allows us to characterize settings where some of the communities can be detected while other cannot. 

Our analysis of the community detection problem  leverages a matrix version of the I-MMSE relation~\cite{reeves:2018a}, which both simplifies and generalizes  techniques used in previous work. In particular, the upper bound on the mutual information  in Theorem~\ref{thm:cF0} is a consequence of a novel non-asymptotic inequality that holds under \emph{any} distribution on the community labels.  Many of our techniques can be applied more generally to other high-dimensional inference  problems, including matrix and tensor factorization.

\subsection{Overview of Approach}
This paper introduces a multivariate measure of performance, which we refer to as the MMSE matrix: 
\begin{align}
\MMSE(\bX \mid \bG)  &  \triangleq \frac{1}{n} \sum_{i=1}^n \bEx_{\bG} \left[  \cov(X_i \mid  \bG) \right]. \label{eq:mmse_matrix}
\end{align}
In this expression, $\cov(X_i \mid  \bG)$ is the covariance matrix of the $i$-th node's label after is has been embedded in to an $\ell$-dimensional Euclidean space (where $\ell$ is either $k$ or $k-1$). 
We show that the MMSE matrix provides important geometrical information about the uncertainty in the community memberships. While the trace of the MMSE matrix corresponds to standard measures of performance  such as the average overlap, the information provided by individual entries in the MMSE matrix can be used to answer more nuanced questions about which of the community relationships can  (or cannot) be recovered.

One of the key ideas in this paper is to focus on community detection in the setting where there is additional covariate information about the labels. Specifically, we assume that one has side-information from the signal-plus-noise model:
\begin{align}
\bY = \bX S^{1/2} + \bN, \label{eq:signal_plus_noise}
\end{align}
where $S$ is an $\ell \times \ell$ positive semidefinite matrix, known as the matrix SNR,  and $\bN$ is an $n \times \ell$ matrix with  i.i.d.\ standard Gaussian entries.

The introduction of the signal-plus-noise model plays an important role both for our analysis and for our interpretation of the results. For example, it allows us to leverage the matrix I-MMSE relation \cite{reeves:2018a} to characterize the MMSE matrix in terms of the gradient of the mutual information: 
\begin{align}
\nabla_S I(\bX; \bG, \bY) & = \frac{n}{2}  \MMSE(\bX \mid \bG, \bY).
\end{align}
Remarkably, this relationship holds generally for any joint distribution on the pair $(\bX, \bG)$. Notice  that the matrix MMSE in \eqref{eq:mmse_matrix} is obtained by evaluating this expression at  $S = 0$. 

The signal-plus-noise model also provides a natural way to address non-identifiability issues that  arise when the distribution over the labels is invariant to permutations.  The key idea is that in the large-$n$ limit, {\it an arbitrarily small amount of side-information is sufficient to break the symmetry in the model}.
Hence, focusing on the double limit
\[
\lim_{S \to 0} \lim_{n\to \infty} \MMSE(\bX \mid \bG, \bY),
\]
provides a meaningful and interpretable measure of average performance that bypasses the need to optimize over an equivalence class of permutations.

Section~\ref{sec:formulas}   provides  formulas for the per-vertex mutual information and MMSE matrix in the large-$n$ limit. These formulas are stated for a degree-balanced stochastic block model and can be approximated numerically with arbitrary precision. Numerical simulations are provided in Section~\ref{sec:analysis}.

\subsection{Notation}
We use  $\mathbb{S}^d$, $\mathbb{S}_+^d$ to denote the space $d \times d$ symmetric matrices and symmetric positive semi-definite matrices, respectively. Given a symmetric positive semi-definite matrix $S$, we use $S^{1/2}$ to denote the unique positive semi-definite square root. Given matrix $A, B \in \mathbb{S}^d$, the relation $A \preceq B$ means that $B - A \in \mathbb{S}_+^d$.

\section{Definitions}
\label{sec:mmse_matrix}

The $k$ community stochastic blockmodel is frequently parameterized in terms of the tuple $(n, p, Q)$ where $p = (p_1, \dots, p_k)$ is a distribution over $k$ communities and $Q \in [0,1]^{k \times k}$ is a symmetric matrix such that $Q_{ab}$ is the probability of an edge between nodes in communities $a$ and $b$.
Without loss of generality, the community labels can be embedded into finite dimensional Euclidean space. Two useful representations  are considered in Sections~\ref{sec:standard} and \ref{sec:whitened_basis}. In Section~\ref{sec:db_sbm} we introduce the degree balanced SBM for which we state the remainder of the results in the paper. Lastly, in Section~\ref{sec:signal_plus_noise} we introduce the signal plus noise problem which we leverage to derive the results for community detection.

\subsection{Standard Basis Representation}\label{sec:standard}

 A natural embedding associates the labels with the standard basis vectors $\{e_1, \dots, e_k\}$ in $\reals^k$, i.e., the columns of the identity matrix. Under this representation,  the expected value of a label vector $X_i$ is a point on the probability simplex.  The conditional covariance is defined by
 \begin{align*}
 \cov(X_i \mid \bG ) & \triangleq  \bEx_{\bX \mid \bG} \left[  \left( X_i - \ex{ X_i | \bG } \right)^T  \left( X_i - \ex{ X_i | \bG } \right)  \right],
 \end{align*}
and the MMSE matrix is defined according to \eqref{eq:mmse_matrix}. By the data processing inequality for MMSE, this matrix satisfies
\begin{align*}
0 \preceq \MMSE(\bX \mid \bG)  \preceq  \MMSE(\bX ) \triangleq \frac{1}{n} \sum_{i=1}^n \cov(X_i). 
\end{align*}
As a consequence, the difference between the MMSE matrix and covariance provides a measure of the difference between the prior and posterior marginals of the labels. 
 
\begin{prop}\label{prop:MMSE_standard} Under the standard basis representation, the $k \times k$ MMSE matrix satisfies 
\begin{align*}
 \gtr\left( \MMSE(\bX ) - \MMSE(\bX \mid \bG  ) \right) 
  = \frac{1}{n} \sum_{i=1}^n  \bEx_{\bG} \left[  \left\| P_{X_i \mid \bG}(\cdot \mid G)   -  P_{X_i}(\cdot) \right\|_{2}^2\right].
\end{align*}
\end{prop}
\begin{proof}
For each $i$, we can write
\begin{align*}
\gtr\left( \cov(X_i) - \ex{ \cov(X_i \mid \bG)} \right) 
 = \ex{ \left\| \ex{ X_i \mid \bG} - \ex{ X_i}  \right\|^2} = \ex{ \left\|P_{X_i \mid \bG}(\cdot \mid G)   -  P_{X_i}(\cdot)  \right\|^2},
\end{align*}
where the first equality follows from the law of total variance and the last step holds because, under the standard bases representation, we have $\ex{ X_{i\ell} \mid \bG } = \pr{ X_{i \ell } = e_{\ell} \mid \bG}$. Summing over all $i$ and normalizing by $n$ completes the proof. 
\end{proof}

Furthermore, the individual entries of the MMSE matrix also provide information about different recovery tasks. For example, consider the problem of determining whether a label belongs to a subset $A \subset [k]$. If we define $\one_A = \sum_{\ell \in A} e_\ell$, then  $\one_A^T X_i$ is binary random variable indicating whether the $i$-th label belongs to $A$. 
Summing the entries in the MMSE matrix indexed by the set $A$ provides a measures of the average error probability: 
\begin{align*}
\one_A^T \MMSE(\bX \mid \bG) \one_A 
 &= \frac{1}{n} \sum_{i=1}^n \bEx_{G} \left[  \var( \one_A^T X_i \mid  \bG)\right].
\end{align*}

\subsection{Whitened Representation} 
\label{sec:whitened_basis} 

Next, we focus on the setting where the labels are identically distributed with probability vector $p = (p_1 , \dots, p_k)$.   The whitened representation is defined to be of a set of $k$ points $\{\mu_1, \dots, \mu_k\}$ in $\reals^{k-1}$ with the property that 
\[
\sum_\ell p_\ell \mu_\ell = 0, \qquad \sum_\ell p_\ell \mu_\ell  \mu_\ell^T= I_{k-1}. 
\]
Under the whitened representation, each label vector has zero mean and identity covariance and thus the MMSE matrix satisfies $0 \preceq \MMSE(\bX \mid \bG)  \preceq   I_{k-1} $.

\begin{remark}[Unique Specification of Whitened Representation]\label{rem:unique_contruction}
The whitened representation can be defined explicitly as a function of $p$ as follows. 
Let $\tilde{p} = (\sqrt{p_1}, \dots, \sqrt{p_k})^T$ and apply the Gram-Schmidt process to the vectors $\{\tilde{p}, e_1, \dots, e_{k-1}\}$ to obtain an orthonormal basis for $\reals^k$ of the form $[\tilde{p}, B]$ where $B$ is $k \times (k-1)$. Then,  the support of the whitened representation is related to the standard basis vectors according to 
\begin{align}
\mu_\ell &= B^T P^{-1/2} e_\ell \quad  \iff \quad  e_{\ell}  = p + P^{1/2} B \mu_\ell  ,\label{eq:representation_transform} 
\end{align}
where $P  = \diag(p)$.  This  construction is unique and has the useful property that 
$\mu_\ell$ lies in the span of $\{e_1, \dots, e_\ell\}$. 
\end{remark}

\begin{prop}\label{prop:MMSE_white}
If the labels are identically distributed then the $(k-1)\times (k-1)$  MMSE matrix of the whitened representation satisfies 
\begin{align*}
\gtr\left( I -  \MMSE(\bX \mid \bG)  \right) 
&  =  \frac{1}{n} \sum_{i=1}^n   \chi^2\left(  P_{X_i , \bG} \,  \|  \,    P_{X_i}  P_{\bG} \right),
\end{align*}
where $\chi^2(P\,  \| \,  Q) = \int ( \dd P/ \dd Q)^2\,  \dd Q $ denotes the chi-squared divergence. 
\end{prop}
\begin{proof}
Noting that $\MMSE(\bX )  = I$ and using the same approach as in the proof of Proposition~\ref{prop:MMSE_standard}, we have
\begin{align}
 \gtr\left(I - \MMSE(\bX \mid \bG  ) \right) 
& = \frac{1}{n} \sum_{i=1}^n  \bEx\left[\left \| \ex{ X_i \mid \bG} - \ex{ X_i}  \right\|_2^2\right]. \label{eq:total_var_white}
\end{align}
Next, let $\tilde{X}_i$ denote the representation of $X_i$ in the standard basis and observe that 
\begin{align*}
\left \|   \ex{ X_i \mid \bG} - \ex{ X_i}   \right\|_2^2
 & =  \left \|  B ^T P^{-1/2}  \ex{ \tilde{X}_i \mid \bG} -  B^T P^{-1/2}  \ex{ \tilde{X}_i}   \right\|_2^2\\
& =  \sum_{\ell = 1}^k \left(\frac{1}{   \sqrt{p_\ell }}  \pr{ \tilde{X}_i = e_\ell \mid \bG }   -  \sqrt{ p_\ell }\right)^2\\
&  =   \chi^2\left(  P_{X_i \mid \bG}(\cdot \mid \bG)  \, \| \,    P_{X_i}(\cdot) \right),
\end{align*}
where we have used \eqref{eq:representation_transform} and the fact that $\ex{ \tilde{X}_i} = p$. Plugging this expression back into \eqref{eq:total_var_white} gives the stated result. 
\end{proof}

For the purposes of analysis, the two representations described above are equivalent in the sense that there is a one-to-one mapping between the $k \times k$ MMSE matrix defined under the standard basis representation and the $(k-1) \times (k-1)$ MMSE matrix defined under the whitened representation. For notational convenience we work in the whitened representation.

\subsection{Degree-Balanced SBM}\label{sec:db_sbm} 
The average degree of an SBM corresponds to the expected number of edges for a node chosen uniformly at random and is denoted  by $d$. An SBM is said to be \emph{degree-balanced} if the expected degree of a node does not depend on its community assignments. This condition is equivalent to saying that $Q p$ is proportional to the all ones vector. 

For the purposes of this paper, it is useful to consider a reparameterization of the degree-balanced SBM in terms of the tuple $(n,d,p, R)$ where $d$ is the average degree and  $R \in \mathbb{S}^{k-1}$. Using this parameterization, the entries of $Q$ are given by
\begin{align}
Q_{ab} = \frac{d}{n} + \frac{\sqrt{d ( 1- d/n) }}{n} \mu_a^T R \mu_b, \label{eq:Qab}
\end{align}
where $\{\mu_1, \dots, \mu_k\}$ are defined as a function of $p$ using the procedure described in Remark~\ref{rem:unique_contruction}. The tuple $(n,d,p,R)$ is valid only if the entries of $Q$ are between zero and one.

The matrix $R$ quantifies the relative strength of relationships between different communities. The eigenvalue decomposition is given by
\begin{align*}
R = U \diag(\lambda) U^T,
\end{align*}
where  $\lambda = (\lambda_1, \dots , \lambda_{k-1})$ are real numbers. To simplify the analysis, we will assume throughout that all the eigenvalues are nonzero so that $R$ is invertible.

We remark that the definition of signal-to-noise ratio given by Abbe and Sandon \cite[Section 2.1]{abbe:2018a} corresponds to $\max_{i} \lambda_i^2$. Furthermore, for the special case of $k=2$ communities, the representation of  $X_i$ is one-dimensional and the formulation of Lelarge and  Miolane~\cite{lelarge:2018} is equivalent to ours.

\subsection{Signal-Plus-Noise Problem} 
\label{sec:signal_plus_noise}

Our analysis uses properties of the signal-plus-noise model given in \eqref{eq:signal_plus_noise}. Throughout this section we will assume the labels are drawn i.i.d.\ according to a probability vector $p = (p_1, \dots , p_k)$ with strictly positive entries and are supported on the whitened representation described in Section~\ref{sec:whitened_basis}. For each $S \in \mathbb{S}_+^{k-1}$,  the task of recovering $\bX$ from $\bY$ decouples into $n$ independent copies of the problem  
\begin{align*}
Y =S^{1/2} X  + N,
\end{align*}
where $X$ is supported on $\{ \mu_1, \dots, \mu_k\}$ with probability vector $p$  and  $N \sim \normal(0, I)$ is independent Gaussian noise. 

Following~\cite{reeves:2018a} we define the the  mutual information function $I_{X} : \mathbb{S}_+^{k-1} \to [0, \infty)$ and matrix-valued MMSE function $M_{X} : \mathbb{S}_+^{k-1} \to \mathbb{S}_+^{k-1} $ according to
\begin{align}
I_{X}(S) & =  I(X ; Y) \label{eq:I_X} \\
M_{X}(S) & = \ex{ \cov(X \mid  Y)}. \label{eq:M_X}
\end{align}
The gradient and Hessian of  $I_{X}(S)$ are given by~\cite[Lemma~4]{reeves:2018a} 
\begin{align}
\nabla_S I_{X}(S) &= \frac{1}{2} M_{X}(S) \\
\nabla^2_S I_{X}(S) & = - \frac{1}{2} \ex{ \cov(X \mid  Y) \otimes \cov(X \mid  Y)}  \label{eq:I_Xpp},
\end{align} 
where $\otimes$ denotes the Kronecker product.
We note that these functions can be approximated using numerical integration methods or Monte-Carlo sampling.

\section{Formulas for Mutual Information and MMSE}\label{sec:formulas} 

Our analysis focuses on a sequence of degree-balanced SBMs where the parameters $(p, R)$ are fixed as the size of the network $n$ scales to infinity.  Additionally, we make two assumptions.

\begin{assumption}[Diverging Average Degree]\label{assump:lin_deg} The average degree of the network $d$ increases with $n$ such that both $d$ and $(n-d)$ tend to infinity. 
\end{assumption}

\begin{assumption}[Definite Matrix]\label{assump:defR} The matrix $R$ is either positive definite or negative definite.  
\end{assumption}

Our first result is stated in terms of the potential function $\cF : \mathbb{S}_+^{k-1} \to \reals_+$ defined by
\begin{align}
\cF(\Delta) & =   I_{X}(\Delta)  + \frac{1}{ 4} \gtr\left( \left( R-  R^{-1} \Delta \right)^2 \right). \label{eq:cF}
\end{align}
where $I_{X}(\cdot)$ is defined by \eqref{eq:I_X}.  Notice that the first term in the potential function is defined exclusively by the prior distribution  of labels $p$ whereas the second term is defined exclusively by the matrix $R$.  By the matrix I-MMSE relation \cite{reeves:2018a}, it can be verified that every stationary point of $\cF(\Delta)$ satisfies the fixed-point equation
\begin{align}
M_{\wX}(\Delta) & = I - R^{-1} \Delta  R^{-1}. \label{eq:FP}
\end{align}
where $M_{X}(\cdot)$ is defined by \eqref{eq:M_X}. Noting that $M_{X}(0) = I$,  we see that  $\Delta = 0$ is always a stationary point. Furthermore,  every solution of \eqref{eq:FP} belongs to the set  $\{ \Delta \, : \, 0 \preceq \Delta \preceq R^2\}$.

\begin{theorem}\label{thm:cF} Under Assumptions~\ref{assump:lin_deg} and \ref{assump:defR}, 
\begin{align*}
\lim_{n \to \infty } \frac{1}{n} I(\bX;\bG)  =   \min_{\Delta \in \psd^{k-1}} \cF(\Delta),
\end{align*}
where $\cF(\Delta)$ is given in \eqref{eq:cF}.
\end{theorem}  

The next result provides an upper bound on the mutual information in the setting where side information is generated according to the signal-plus-noise model \eqref{eq:signal_plus_noise} parameterized by a positive semi-definite matrix  $S$.  To characterize this setting, we define the modified potential function:
\begin{align}
\cF(\Delta, S) & =   I_{X}(S+ \Delta)  + \frac{1}{ 4} \gtr\left( \left( R-  R^{-1} \Delta \right)^2 \right). \label{eq:cF0}
\end{align}
Notice that the main difference from \eqref{eq:FP} is that the side information changes the prior information about the labels. 

\begin{theorem}\label{thm:cF0} Suppose that $\bY$ is generated according to the signal-plus-noise model  \eqref{eq:signal_plus_noise} with matrix $S \in \mathbb{S}_+^{k-1}$. Under Assumption~\ref{assump:lin_deg}, 
\begin{align*}
\limsup_{n \to \infty } \frac{1}{n} I(\bX;\bG, \bY )  \le   \min_{\Delta \in \psd^{k-1}} \cF(\Delta, S) .
\end{align*}
 where $\cF(\Delta, S)$ is given in \eqref{eq:cF0}. 
\end{theorem}

\begin{remark}
Similar to previous work \cite{deshpande:2017,caltagirone:2018, lelarge:2018,  krzakala:2016 ,barbier:2016a,lesieur:2017}, our proofs of Theorems~\ref{thm:cF} and \ref{thm:cF0} use a channel universality argument to relate the community detection problem to a low-rank estimation problem. Assumption~2 is needed for the proof of Theorem~1, which leverages \cite[Theorem~12]{lelarge:2018}.  To prove Theorem \ref{thm:cF0} we develop a novel variation of the Guerra interpolation method that exploits the matrix I-MMSE relationship~\cite{reeves:2018a} to provide a general and non-asymptotic upper bound.
\end{remark}

Next, we recall that that by the data processing inequality, the MMSE matrix satisfies
\[
\MMSE(\bX \mid \bG) \succeq \MMSE(\bX \mid \bG, \bY),
\]
for all $S \in \mathbb{S}^{k-1}_+$.  For any fixed problem size $n$,  the difference between these matrices converges to zero as $S \to 0$. However, in the large-$n$ limit it is possible that the limiting behavior is discontinuous with respect to $S$. This can occur, for example, when the SBM is invariant to permutations of the labels and hence $\MMSE(\bX \mid \bG) = \MMSE(\bX)$. The presence of side-information with an arbitrarily small positive definite matrix $S$ is sufficient to break the permutation invariance, and thus  the small-$S$ limit provides a meaningful measure of recovery performance that overcomes the non-identifiability  issues. 

The following result follows from the  matrix  I-MMSE relation and Theorems~\ref{thm:cF} and \ref{thm:cF0}. \iflongpaper The proof is given in Appendix~\ref{sec:thm:MMSE_UB_proof}.\fi

\begin{theorem}\label{thm:MMSE_UB} Consider Assumptions~\ref{assump:lin_deg} and \ref{assump:defR}. For every $S \succ 0$,
\begin{align*}
  \limsup_{n \to \infty } \lambda_\mathrm{max} \left(  \MMSE(\bX \mid \bG, \bY) - M_{X}(\Delta^*)  \right)  \le  0
\end{align*}
where $\Delta^*$ denotes any minimizer of $\cF(\Delta)$. In other words, 
\[
\MMSE(\bX \mid \bG, \bY) \preceq  M_{X}( \Delta^*)  + o_n(1),
\]
where $o_n(1)$  denotes a sequence of symmetric matrices that converges to zero as $n \to \infty$. 
\end{theorem}

The numerical experiments of Section~\ref{sec:analysis} suggest that the upper bounds in Theorem~\ref{thm:cF0} are asymptotically tight, {\it i.e.}, that the MMSE matrix satisfies
\[
\MMSE(\bX \mid \bG, \bY) = M_{X}(S + \Delta^*) + o_n(1) 
\]
for almost all $S$, where $\Delta^*$ is the unique minimizer of $\cF(\cdot, S)$.

The next result provides an asymptotic lower bound on the problem of estimating $\bX R \bX^T$, which implies a lower bound on $\MMSE(\bX \mid \bG)$. The proof is given in Appendix~\ref{proof:thm:MMSE_LB}. 

\begin{theorem}\label{thm:MMSE_LB} Under Assumptions~\ref{assump:lin_deg} and \ref{assump:defR}, 
\begin{align*}
\liminf_{n \to \infty} \frac{1}{n^2 } \ex{ \left\| \bX R \bX^T - \ex{ \bX R  \bX^T \mid \bG } \right \|_F^2}   & \ge  \min_{\Delta \in \cD} \gtr (R^2  - R^{-2} \Delta^2  ) , 
\end{align*}
where $\cD = \arg \min \cF(\Delta)$. Furthermore, this implies that
\begin{align*}
\liminf_{n \to \infty}  \gtr\left( R^2 (I - \MMSE(\bX \mid \bG))^2 \right)  & \ge  \min_{\Delta \in \cD}   \gtr\left( R^2 (I - M_X(\Delta))^2 \right) .
\end{align*}
\end{theorem}

\color{black}

\section{Implications for Weak Recovery}
\label{sec:weak_recovery}

Broadly speaking, weak recovery refers to the  ability to produce an estimate that is positively correlated with the ground truth. In the context of community detection, the precise definition of weak recovery is a bit more nuanced due to the fact that symmetries in the problem formulation can result in a posterior distribution that is invariant to  permutations of the labels. As a specific example, consider the two-community degree-balanced SBM where each community is equally likely. Even if an estimator can partition the nodes into two groups such that all of the nodes in each group belong to the same community,  it is impossible to determine which label should be assigned to which group. 

One approach that is taken in the literature to address this nonidentifiability assesses the performance of an estimator after choosing a permutation of the labels that leads to the best performance; see e.g., \cite[Section~2]{abbe:2018a}. Another approach focuses on the related problem of estimating the pairwise interaction terms $\{ X^T_i R X_j \}$. Specifically weak recovery with respect to the pairwise interactions is possible if 
\begin{align}
\limsup_{n  \to \infty} \frac{1}{n^2  } \sum_{i, j} \MMSE( X_i^T R X_j \mid \bG) < \frac{1}{n^2  } \sum_{i, j} \var(X_i^T R X_j ), \label{eq:weak_recovery_mtx} 
\end{align}
where $\MMSE( X_i^T R X_j \mid \bG) \triangleq \bEx_{\bG}\left[ \var( X_i^T R X_j \mid \bG)\right]$.  Notice that under the whitened basis representation we propose, $\var(X_i^T R X_j ) = \|R\|_F^2$ and this condition is equivalent to
\begin{align*}
\limsup_{n  \to \infty} \frac{1}{n^2 } \ex{ \left\| \bX R \bX^T - \ex{ \bX R  \bX^T \mid \bG } \right \|_F^2}  <  \|R\|_F^2
\end{align*}

Following the approach taken in this paper, we see that a natural alternative is to focus on the small-$S$ behavior of the MMSE matrix. In particular, we say that weak recovery is possible if
\begin{align}
\inf_{S \succ 0}  \liminf_{n \to \infty } \gtr\left(  \MMSE(\bX \mid \bG, \bY)  \right) < \gtr \left(  \MMSE(\bX) \right )  . \label{eq:weak_recovery}
\end{align}

In view of these definitions, we see that Theorem~\ref{thm:MMSE_UB}  and Theorem~\ref{thm:MMSE_LB} provide necessary and sufficient  conditions for  weak recovery, depending on whether the potential function $\cF(\cdot)$ has a unique minimizer at zero.  

\begin{theorem}[Weak Recovery]\label{thm:weak_recovery} 
Consider Assumptions~\ref{assump:lin_deg} and \ref{assump:defR}. If $\cF(\cdot)$ has a  minimizer that is not equal to zero  then weak recovery in the sense of  \eqref{eq:weak_recovery} is possible.  Conversely, if  $\cF(\cdot)$ has a unique minimizer at zero, then weak recovery in the sense of  \eqref{eq:weak_recovery_mtx} is not possible. \color{black} 
\end{theorem}

\color{black}

Evaluating the Hessian of the potential function at zero provides a simple test to determine whether $\Delta=0$ is a local minimum. Using \eqref{eq:I_Xpp}, it can be shown that 
\begin{equation*}
\nabla^2 \cF(\Delta) \Big \vert_{\Delta =0}  \propto   R^{-1} \otimes R^{-1}  - I_{(k-1)^2}. 
\end{equation*}
Therefore, if $\max_{i} \lambda^2_i(R) > 1$ then $\Delta =0$ is not a local minimizer.

\section{Numerical Experiments} 
\label{sec:analysis}

\begin{figure*}
        \centering
        \begin{subfigure}[b]{0.44\textwidth}
                \centering
                \begin{picture}(200,200)
               \put(0,0) {\includegraphics[width=\textwidth]{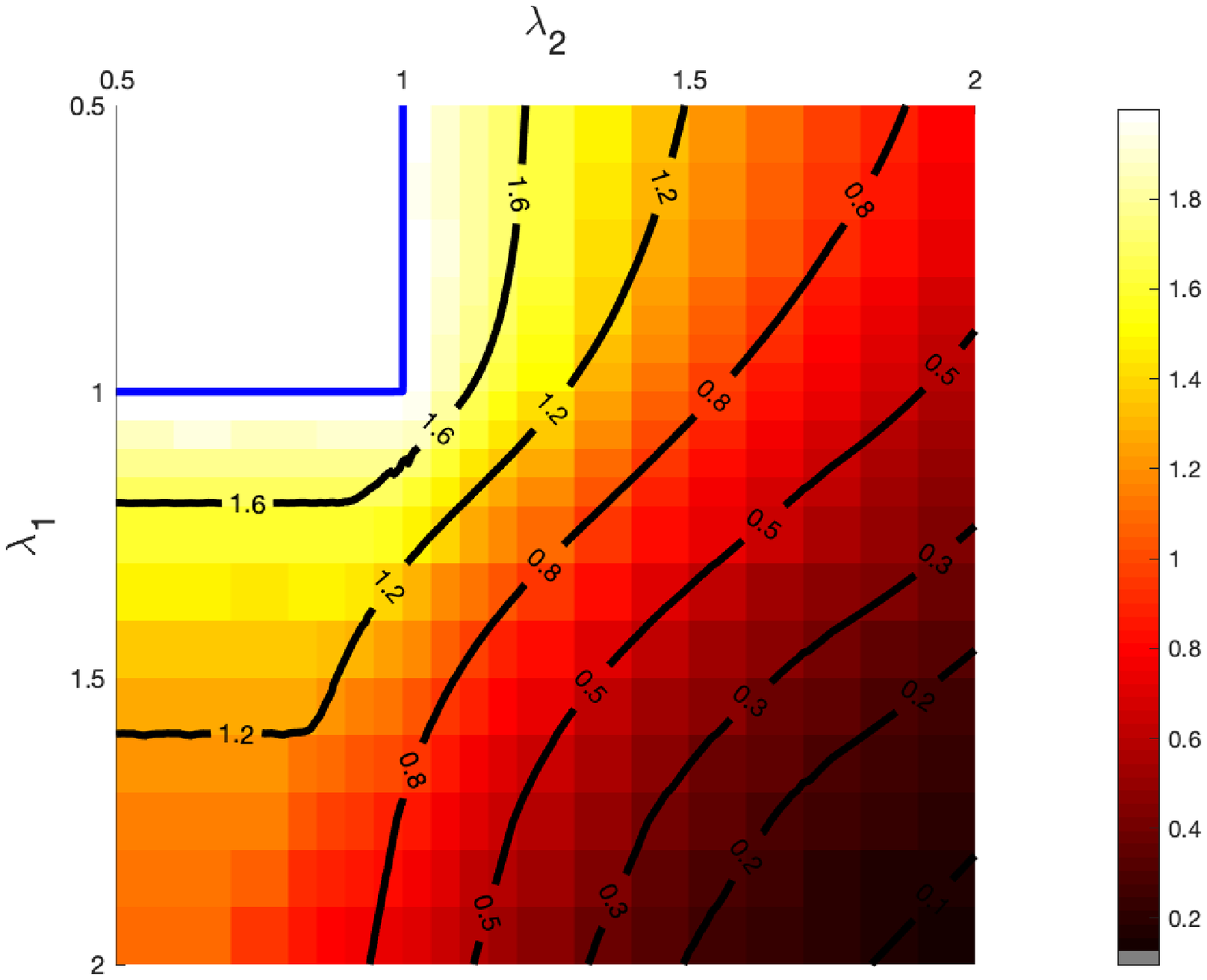}}
              \put(41,135){\footnotesize{\textcolor{blue}{$\Delta^* = 0$}}}               
             \end{picture}
             \vspace*{-.4cm}
                \caption{$p = (1/3, 1/3, 1/3)$ }
                \label{fig:grid1}
        \end{subfigure}%
        \begin{subfigure}[b]{0.44\textwidth}
                \centering
                \begin{picture}(200,200)
               \put(0,0) {\includegraphics[width=\textwidth]{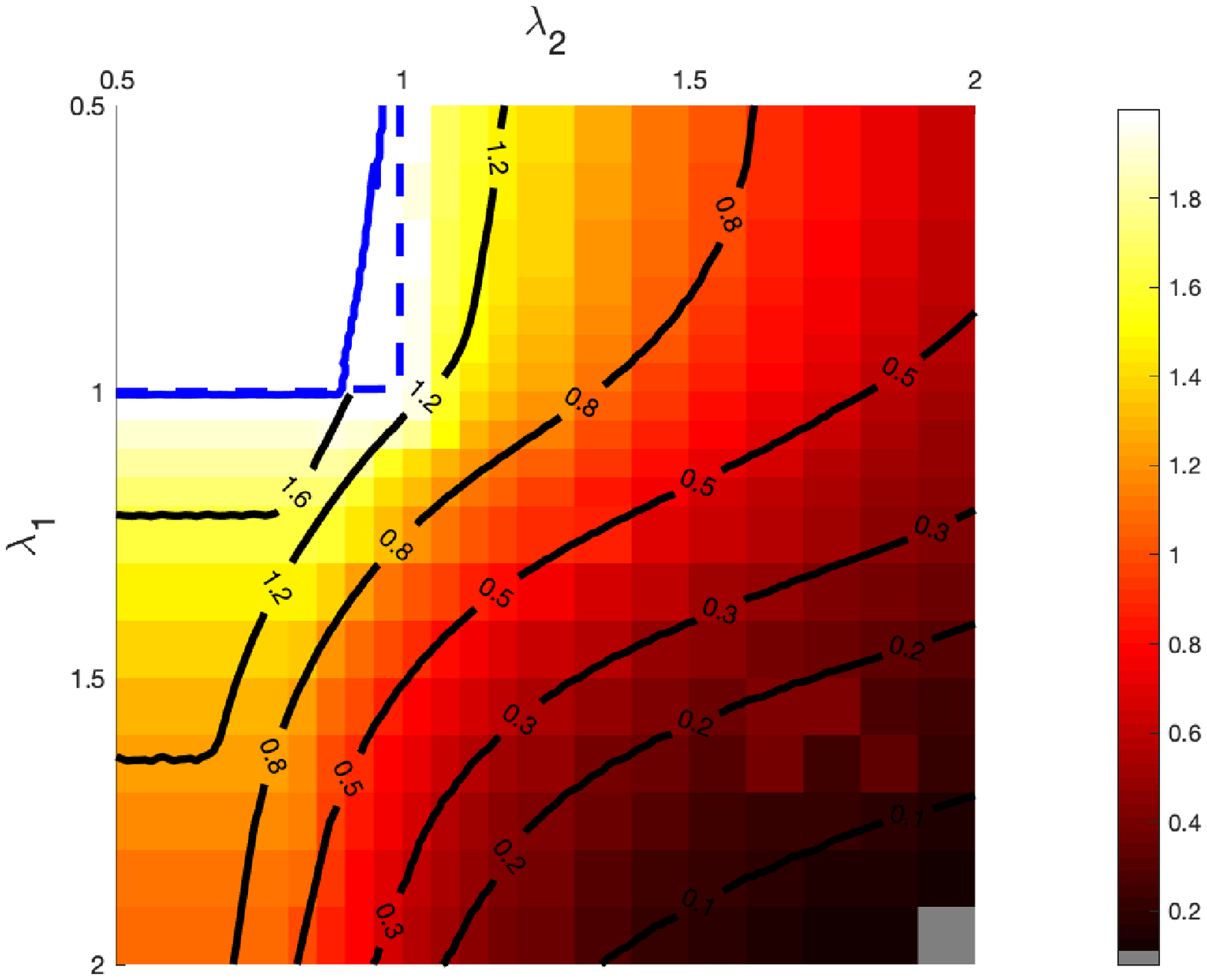}}
              \put(41,135){\footnotesize{\textcolor{blue}{$\Delta^* = 0$}}}
               \end{picture}
                     \vspace*{-.4cm}
                \caption{$p = (0.6, 0.3, 0.1)$ }
                \label{fig:grid2}
        \end{subfigure}
 \caption{\small 
 \textcolor{black}{Comparison of upper bound on $\gtr(\MMSE(\bX \mid \bG))$  given in Theorem~\ref{thm:MMSE_UB} (black contour lines) and the empirical MSE of  belief propagation (heat map) on a network of size $n = 10^5$ with average degree $d=30$.  In both cases,  $R = \diag(\lambda_1, \lambda_2)$. The upper bound on the weak recovery threshold given in Theorem~\ref{thm:weak_recovery} (solid blue line) corresponds to the boundary where $\Delta^* = 0$. The weak recovery threshold for acyclic BP \cite{abbe:2018a} (dashed blue line) corresponds to $\max(\lambda_1, \lambda_2) = 1$. The grey region in (b) corresponds to settings where $(n,d,p,R)$ does not define a valid SBM.}
 }        \label{fig:grid}
\end{figure*}

This section compares the asymptotic bounds given in  Section~\ref{sec:formulas} with the MSE obtained using belief propagation (BP). The case of the three-community degree balanced SBM $(n,d,p, R)$  is illustrated in Figure~\ref{fig:grid}. The black contour lines correspond to the trace of $M_{X}(\Delta^*)$ where $\Delta^*$ is the global minimizer of the potential function defined in \eqref{eq:cF}.  The heat map values correspond to the empirical MSE of  the BP algorithm described in \cite{decelle:2011a} applied to a network of size $n = 10^5$ with average degree $d = 30$. Each pixel is the median of eight independent trials and the MSE is measured with respect to the whitened basis representation. In each trial, the BP algorithm is run using fifteen different random initializations and the MSE is assessed based on the initialization that produces in the lowest predicted MSE. 

In the case of uniform community assignments (Figure~\ref{fig:grid1}), the weak recovery limit for acyclic BP \cite{abbe:2018a} is equal to our upper bound on the weak detection threshold. Furthermore, we see that there is a close correspondence between the asymptotic formula and the empirical results. Note that the  special case $\lambda_1=\lambda_2$ corresponds to the three-community symmetric SBM.

In the case of non-uniform  community assignments (Figure~\ref{fig:grid2}), there exists a region of the parameter space where weak recovery is possible with $\max(\lambda_1, \lambda_2)< 1$. The existence of such a region has been shown previously in the special case of the two-community asymmetric SBM~\cite{caltagirone:2018}. We also see that the asymptotic formulas match the empirical behavior qualitatively, although the empirical MSE is worse than is suggested by the formulas.  The grey region in Figure~\ref{fig:grid2} corresponds to settings where $(n,d,p,R)$ does not define a valid SBM.

\subsubsection*{Numerical Approximation of Formulas} 

We use Monte Carlo sampling  to approximately evaluate the functions $I_{X}$ and $M_{X}$, and we use the concave-convex procedure \cite{yuille:2002} to explore the local minima of the potential function.  Starting is an initialization point $\Delta^0$, a sequence of iterates is obtained according to 
\begin{align*}
\Delta^{t+1} &=  (1 - \eps)  \left( R^2 - RM_{X}( \Delta^t)R \right)  +  \eps \Delta^t,
\end{align*}
where $\eps \in [0, 1)$ is a dampening parameter.

\section{Main Steps in Proof} 

This section provides an overview of the main theoretical results of the paper. These results are described in the context of a more general inference problem where the goal is to  estimate a random $n \times \ell$ matrix $\bX = [ X_1, \dots, X_n]^T$. The setting of the $k$-community degree-balanced SBM described in Section~\ref{sec:formulas} corresponds to the special case where $\ell = k-1$ and the rows of $\bX$ are drawn i.i.d.\ from the whitened distribution described in Section~\ref{sec:whitened_basis}. 

\subsection{Equivalence between Observation Models} 
The high-level idea behind our approach is to established an equivalence between three different observations models. The first observation model is the signal-plus-noise model given by:
\begin{align}
\bY = \bX S^{1/2} + \bN, \label{eq:obs1} 
\end{align}
where $S \in \psd^\ell$ and $\bN$ is an $n \times \ell$ standard Gaussian matrix, i.e., the entries are i.i.d.\ $\normal(0,1)$.

To describe the second observation model, we first define the symmetric $n \times n$ random matrix
\begin{align}
\bW = \frac{1}{\sqrt{n}} \bX R \bX^T. \label{eq:W} 
\end{align}
where $R \in \mathbb{S}^\ell$. Then, the observations are given by 
\begin{align}
\bZ = \sqrt{t} \bW + \bm{\xi}, \label{eq:sym_model}
\end{align}
where $t \in [0, \infty)$ and $\bm{\xi}$ is an $n \times n$ standard Gaussian Wigner matrix, i.e. a symmetric matrix whose entries above the diagonal are i.i.d.\ $\normal(0, 1)$ and whose entries on the diagonal are i.i.d.\ $\normal(0, 2)$. 

For the last model, the observations consist of an $n$-node simple graph, which is represented by its adjacency matrix $\bG \in \{0,1\}^{n \times n}$. By convention the diagonal entries are set to zero and the off-diagonal entries are given by $G_{ij} =G_{ji}=1$ if there is an edge between nodes $i$ and $j$ and zero otherwise.  Our results apply to the setting where the entries of the adjacency matrix are drawn independently conditional on $\bW$ according to 
\begin{align}
G_{ij} \sim \text{Bernoulli}\left( \frac{d}{n}  + \sqrt{\frac{d}{n}  \left(1- \frac{d}{n} \right ) } W_{ij} \right) , \quad i < j,\label{eq:Aij}
\end{align} 
where $d \in (0, n)$ parameterizes the expected number of edges.

Notice that both \eqref{eq:sym_model} and \eqref{eq:Aij} consist of elementwise observations of $\bW$ from a fixed output channel. The following result provides a link between the mutual information in these observation models. The proof is given in Appendix~\ref{proof:thm:GtoZ}.

\begin{theorem}[Channel Universality]\label{thm:GtoZ} 
Let $\bW$ be a symmetric $n \times n$ random matrix with bounded entries  $|W_{ij}|  \le B/\sqrt{n}$ and finite support of cardinality $N$.  Let $\bZ$ be drawn according to \eqref{eq:sym_model} with $t = 1$ and $\bG$ be drawn according to \eqref{eq:Aij}. Given any $\delta > 0$, there exists a constant $C(\delta, B)$ such that 
\begin{align*}
\left | I(\bW; \bG) - I(\bW ; \bZ) \right| \le   C(\delta, B)  \left( \frac{n^{3/2} + n  \sqrt{ \log N}   }{ \sqrt{d (n-d)}}  + \frac{n  \log N   }{ d (n-d)}   \right),
\end{align*} 
uniformly for all integers $n > \delta / 2$ and $d \in [\delta, n - \delta]$.  
\end{theorem}

\begin{remark} The concept of channel universality appeared in the work of Korada and Montanari~\cite{korada:2011} and subsequently developed in the context of community detection \cite{deshpande:2017, caltagirone:2018, lelarge:2018}  and low-rank matrix estimation \cite{lesieur:2017, barbier:2016a, krzakala:2016}. In relation to this work, the contribution of 
Theorem~\ref{thm:GtoZ} is that it holds under more general assumptions on both $\bW$ and the average degree $d$. \end{remark}

Theorem~\ref{thm:GtoZ}  implies that the joint information in $(\bG,\bY)$ about $\bX$ is asymptotically equivalent to the joint information in $(\bY,\bZ)$ about $\bX$.

\begin{cor}\label{cor:GtoZ} Let $(\bX, \bG)$ be drawn according to the degree-balance SBM with parameters $(n, d,p, R)$ where $p$ and $R$ are fixed and $d$ scales with $n$ such that both $d$ and $(n-d)$ tend to infinity. Let $\bY$ be drawn according to \eqref{eq:obs1}  and let  $\bZ$ be drawn according to  \eqref{eq:sym_model}  with $t = 1$ and  $\bW = n^{-1/2} \bX R \bX^T$. Then, 
\begin{align*}
\lim_{n \to \infty} \frac{1}{ n}  \left | I(\bX;  \bG, \bY)  -  I(\bX;  \bY, \bZ)  \right|   = 0 
\end{align*}
\end{cor}
\begin{proof}
Combining the chain rule for mutual information with the Markov structure in 
 $(\bW,  \bX , \bY, \bZ)$ leads to
\begin{align*}
I(\bX;  \bG, \bY)  -  I(\bX;  \bY, \bZ)  & = I(\bW;  \bG \mid  \bY) - I(\bW;  \bZ \mid  \bY).
\end{align*}
By assumption, $\bX$ has finite support of cardinality $ k^n$ and bounded entries. This implies that $\bW$ has  finite support of cardinality $N = k^n$ and bounded entries $|W_{ij}| \le B/\sqrt{n}$ where the constant $B$ depends only on $(p ,R)$ .  
For every realization $\by$ of $\bY$,  Theorem~\ref{thm:GtoZ}  implies that there is a constant $C(p,R)$ such that 
\[ 
\frac{1}{n} |I(\bW; \bG \mid \bY = \by)  - I(\bW; \bZ  \mid \bY = \by) | \le C(p, R)  \sqrt{ \frac{1}{d} + \frac{1}{ n - d} } ,
\] 
for all $n$ and $d$ sufficiently large.  The stated result then follows from Jensen's inequality and the assumptions on $n$ and $d$.   
\end{proof}

\subsection{Interpolation via Mutual Information} 

Theorem~\ref{thm:GtoZ} provides a link between  community detection and symmetric matrix estimation. The next step in our analysis is to study an interpolating function that transitions smoothly from the symmetric matrix  model to the signal-plus-noise model.  We note that a number of approaches have been developed in the statistical physics literature, including Guerra's interpolation method~\cite{guerra:2003} and the adaptive interpolation method \cite{barbier:2018}. In this paper we consider an approach inspired by the work of Reeves~\cite{reeves:2017e}, which leverages the functional properties of mutual information in Gaussian channels.

The central object of interest is the mutual information functions $I_{\bX, \bW}: \mathbb{S}_+^\ell \times [0, \infty) \to \reals$ defined by
\begin{align}
I_{\bX, \bW}(S,t) & \triangleq \frac{1}{n}   I(\bX, \bW ; \bY, \bZ). \label{eq:I_XW}
\end{align}
This function has a number of useful properties. Combining the chain rule for mutual information with the Markov structure in  $(\bW, \bX , \bY, \bZ)$ allows us to write
\begin{align*}
I(\bX, \bW ; \bY, \bZ ) = I(\bX; \bY)  + I(\bW; \bZ \mid \bY ) \\
 = I(\bW ; \bZ ) + I(\bX ; \bY \mid \bZ). 
\end{align*}
Hence, the special cases $t=0$ and $S=0$ are given by
\begin{align*}
I_{\bX, \bW }(S, 0 ) & = I_{\bX }(S)  \triangleq 
\frac{1}{n}  I(\bX; \bY )\\
I_{\bX, \bW}(0, t ) & = I_{\bW}(t)  \triangleq 
\frac{1}{n}  I(\bW; \bZ).
\end{align*}
In this way, $I_{\bX, \bW}(S,t)$ provides a bridge between the symmetric matrix estimation problem, with or without side information, and the signal-plus-noise problem.  Notice that if the rows of $\bX$ are independent, then  $I(\bX; \bY) = \frac{1}{n} \sum_{i=1}^n I(X_i ; Y_i)$.  In particular,  if the rows $\bX$ are drawn i.i.d.\  from a distribution $P_X$ on $\reals^d$  (as is assumed in Theorem~\ref{thm:cF0}) then $I_{\bX}(S)$ is equal to the mutual information function $I_X(S)$ introduced in Section~\ref{sec:signal_plus_noise}.

It was previously shown that $I_{\bX, \bW}(S,t)$ possesses several desirable properties: it is concave and twice differentiable in the pair $(S,t)$~\cite[Lemma~4]{reeves:2018a}. 
Let the partial gradients with respect to the first and second arguments be denoted by  $I^{(1)}_{\bX,\bW} : \psd^\ell \times [0, \infty) \to \psd^\ell$ and $I^{(2)}_{\bX,\bW}: \psd^\ell \times [0, \infty) \to \reals$, respectively. By the matrix I-MMSE relation, it follows that: 
\begin{align}
 I^{(1)}_{\bX, \bW}(S, t) & = \frac{1}{2} \MMSE(\bX \mid \bY, \bZ) \label{eq:IXW_1}\\
  I^{(2)}_{\bX, \bW}(S, t) & = \frac{1}{4 n} 
   \ex{ \| \bW - \ex{ \bW \mid \bY, \bZ}\|_F^2}.
\end{align}
The details of this derivation are given in Appendix~\ref{sec:I_MMSE_sym}.

The next result provides a non-asymptotic upper bound on  $I_{\bX, \bW}(S,t)$ in terms of the signal-plus-noise model. Remarkably, the only restriction on $\bX$ is that it has finite fourth moments.  The proof is given in Section~\ref{sec:thm:I_XW_UB_proof}.

\begin{theorem}\label{thm:I_XW_UB}
Let $\bX \in \mathbb{R}^{n \times \ell}$ be a random matrix with finite fourth moments and let $\bW = \frac{1}{\sqrt{n}} \bX R \bX^T$ where $R \in \mathbb{S}^\ell$ is invertible.  For all $S \in \mathbb{S}^\ell_+$ and  $t \in (0, \infty)$, the mutual information function defined in \eqref{eq:I_XW}  satisfies
\begin{align*}
I_{\bX, \bW}(S, t)  & \le \inf_{\Delta \in \psd^d}  \left\{ I_{\bX}(S + t \Delta )  +\frac{t}{4}  \gtr\left(  \left(  \Gamma  R -  R^{-1} \Delta \right)^2  \right) \right\} \\
& \quad 
+ \frac{t}{ 4n^2} \ex{  \left \| R  \bX^T \bX -   R \ex{ \bX^T \bX} \right\|_F^2},
\end{align*}
where $\Gamma = \frac{1}{n} \ex{ \bX^T \bX}$. 
\end{theorem}

If the rows of $\bX$ are sufficiently uncorrelated then the term $\frac{1}{n^2} \ex{  \left \| R  \bX^T \bX -   R \ex{ \bX^T \bX} \right\|_F^2}$ converges to zero in the large-$n$ limit. The case of i.i.d.\ rows is summarized as follows: 

\begin{cor}\label{cor:I_XW_UB}
Let $\bX \in \mathbb{R}^{n \times \ell}$ be a random matrix whose rows are drawn i.i.d.\ from a distribution $P_{X}$ on $\reals^d$ with finite forth moments and let $\bW = \frac{1}{\sqrt{n}} \bX R \bX^T$ where $R \in \mathbb{S}^\ell$ is invertible.  For all $S \in \mathbb{S}^\ell_+$ and  $t \in (0, \infty)$, the mutual information function defined in \eqref{eq:I_XW}  satisfies
\begin{align*}
\limsup_{n \to \infty} I_{\bX, \bW}(S, t)  & \le \inf_{\Delta \in \psd^d}  \left\{ I_{\bX}(S + t \Delta )  +\frac{t}{4}  \gtr\left(  \left(  \ex{ X X^T}  R -  R^{-1} \Delta \right)^2  \right) \right\} .
\end{align*}
\end{cor}
\begin{proof}
Noting that $R \bX^T \bX = \sum_{i=1}^n R X_i X_i^T$ is the sum of  $n$ i.i.d.\ matrices leads to 
\[
\frac{1}{n^2} \ex{  \left \| R  \bX^T \bX -   R \ex{ \bX^T \bX} \right\|_F^2} = \frac{1}{n} \ex{  \left \| R  X X^T -   R \ex{ X X^T} \right\|_F^2},\]
which converges to zero as $n$ increases to infinity. 
\end{proof}

Combining Corollary~\ref{cor:GtoZ}  and Corollary~\ref{cor:I_XW_UB} leads directly to an upper bound on the mutual information in the community detection problem (Theorem~\ref{thm:cF0}). The details of the proof are given in Appendix~\ref{sec:thm:cF0_proof}.  To show that this bound is tight requires significantly  more work.  In this direction, we build upon the work of Lelarge and  Miolane~\cite[Theorem~12]{lelarge:2018}, who give an explicit characterization of the large-$n$ limit for the matrix estimation problem in the setting where $S= 0$. Although their result is stated originally for the special case where $R$ is the identity matrix, it extends to the case described below, where $R$ is definite. For completeness a detailed mapping between their statement of this result and the one used in this paper is provided in Appendix~\ref{sec:thm:matrix_factorizatiion_proof}.

\begin{theorem}[{Lelarge and Miolane~\cite[Theorem~12]{lelarge:2018}}]\label{thm:matrix_factorizatiion} 
Let $\bX \in \mathbb{R}^{n \times \ell}$ be a random matrix whose rows are drawn i.i.d.\ from a distribution $P_{X}$ on $\reals^\ell$ with finite second moments and let $\bW = \frac{1}{\sqrt{n}} \bX^T R \bX$ where $R$ is either positive definite or negative definite.  For all  $t \in(0, \infty)$, the mutual information function defined in \eqref{eq:I_XW}  satisfies
\begin{align*}
\lim_{n \to \infty} I_{\bX, \bW}(0, t)  & = \inf_{\Delta \succeq 0}  \left\{ I_{\bX}( t \Delta )  +\frac{t}{4}  \gtr\left(  \left(  \ex{ X X^T}   R -  R^{-1} \Delta \right)^2  \right) \right\} .
\end{align*}
\end{theorem}

\subsection{Proof of Theorem~\ref{thm:I_XW_UB}} \label{sec:thm:I_XW_UB_proof}

The first step in the proof is given by the the following lemma, which establishes a functional  relationship between the first and second partial gradients of $I_{\bX, \bW}(S,t)$. 

\begin{lemma}\label{lem:grad_inq} 
 The gradients of the function $I_{\bX, \bW}(S,t)$ defined in \eqref{eq:I_XW}  satisfy 
\begin{align}
  I^{(2)}_{\bX, \bW}(S, t) \le \frac{1}{4}  g\left( 2 I^{(1)}_{\bX, \bW}(S, t) \right),
 \end{align}
 where $g: \mathbb{S}^\ell_+ \to \reals$ is defined according to
\begin{align}
g(U) & = \frac{1}{n^2} \gtr\left( \ex{  \left( R \bX^T \bX  \right)^2} \right)   -  \gtr\left(  \left(  R \left(\Gamma  -   U \right) \right)^2  \right). \label{eq:g} 
\end{align}
with $\Gamma = \frac{1}{n} \ex{ \bX \bX^T}$. 
\end{lemma}
\begin{proof}
Based on the analysis of the MMSE matrix of a linear Gaussian channel with matrix input (Appendix~\ref{sec:matrix_inputs}) and the partial derivatives of the mutual information function in symmetric matrix estimation (Appendix~\ref{sec:I_MMSE_sym}) we obtain
\begin{align*}
I^{(1)}_{\bX, \bW}(S,t)  & = \frac{1}{2 n } \left(   \ex{ \bX^T \bX}  -   \ex{ \bA^T \bB} \right)  \\
 I^{(2)}_{\bX, \bW}(S,t) & = \frac{1}{4 n^2} \left(   \gtr\left( \ex{  \left( R \bX^T \bX  \right)^2} \right)   -  \gtr\left(  \ex{\left(  R \bA^T \bB \right)^2  }\right)  \right), 
\end{align*}
where $\bA$ and $\bB$  are conditionally independent draws form the posterior distribution of $\bX$ given $ (\bY, \bZ)$.  Comparing these expressions with the definition of $g(U)$, leads to
\begin{align*}
\frac{1}{4} g\left( 2  I^{(1)}_{\bX, \bW}(S,t) \right)  -  I^{(2)}_{\bX, \bW}(S,t)  & = \frac{1}{n^2} \gtr\left(  \ex{\left(  R \bA^T \bB \right)^2  }\right) - \frac{1}{n^2} \gtr\left( \left(  \ex{  R \bA^T \bB}  \right)^2  \right)\\
& = \frac{1}{n^2}  \gtr\left(  \ex{\left(  R \bA^T \bB  - \ex{  R \bA^T \bB}  \right)^2  }\right).
\end{align*}
Noticing that this expression is non-negative completes the proof.
\end{proof}


The next step in our analysis is to focus on the convex conjugate (or Legendre--Fenchel transform) of $I_{\bX, \bW}(\cdot, t)$. Specifically, we define the extended real-valued function $J_{\bX,\bW} : \mathbb{S}_+^\ell  \times [0,t) \to \reals \cup \{+ \infty\}$ according to
\begin{align}
J_{\bX,\bW}(U, t) & \triangleq \sup_{S \in\mathbb{S}^\ell_+} \left\{ I_{\bX, \bW}(S , t ) - \frac{1}{2} \gtr(S U)  \right\}. \label{eq:J_XW}
\end{align}
Here, we have introduced the factor of one half in so that the dual variable $U$ can be associated with the MMSE matrix. 
The function $J_{\bX,\bW}(\cdot, t)$ is convex  because it is the pointwise maximum of affine functions. By the Fenchel--Moreau theorem~(see e.g.,~\cite[Theorem~13.37]{bauschke:2017}), the fact that $I_{\bX, \bW}(\cdot , t)$  is a proper upper-semicontinuous concave function  implies that the Legendre--Fenchel transform is a bijection, and thus 
\begin{align}
I_{\bX,\bW}(S, t) & = \inf_{U \in\cU } \left\{ J_{\bX, \bW}(U , t ) + \frac{1}{2} \gtr(S U)  \right\}, \label{eq:J_biconjugate}
\end{align}
where  $\cU \triangleq \{ 2 I^{(1)}_{\bX}(S)   \, : \,  S \in \psd^\ell  \}\subseteq \psd^\ell.$ 

%

Working with the transformed representation allows us to convert the functional constraint on the partial derivatives given in Lemma~\ref{lem:grad_inq} into an upper bound on the convex conjugate. 

\begin{lemma}\label{lem:J_XW_inq}
For all $U \in \cU$ we have
\begin{align}
 J_{\bX, \bW}(U , t )    \le  J_{\bX}(U ) +  \frac{t}{4} g(U),
\end{align}
where  $g(U)$  is defined in \eqref{eq:g}. 
\end{lemma}
\begin{proof}
The assumption that $U \in \cU$  combined with the fact that $I^{(1)}_{\bX, \bW}(S, \cdot )$ is non-increasing in the Loewner partial order  ensures that supremum with respect to $S$ in \eqref{eq:J_XW} is attained on at least one point $S^*(U, t) \in \mathbb{S}_+^\ell$. By the Karush--Kuhn--Tucker conditions, the gradient with respect to $S$ evaluated at this point satisfies 
\begin{align}
I^{(1)}_{\bX, \bW}(S^*(U, t) , t ) \preceq \frac{1}{2}  U.  \label{eq:MleU}
\end{align}
Next, we note that $g(U)$ is non-decreasing with respect to the Loewner partial order. To see why, observe that for any $0 \preceq U \preceq V \preceq \Gamma$, we have $g(V) - g(U)   = \gtr ( R (V-U) R (2 \Gamma - U - V ) )  \ge 0$.  

We now employ the envelope  theorem~\cite{milgrom:2002},  which implies that $ J_{\bX, \bW}(U , t )$ is absolutely continuous in $t$ with
\begin{align}
J_{\bX, \bW}(U , t ) - J_{\bX, \bW}(U , 0 )  = \int_0^t  I^{(2)}_{\bX, \bW}(S^*(U,\tau )  , \tau ) \, \dd \tau.\label{eq:J2_UB}
\end{align}
The integrand in this expression can be upper bounded as follows: 
\begin{align}
  I^{(2)}_{\bX, \bW}(S^*(U,t)  , t )   \le  \frac{1}{4} g\left( I^{(1)}_{\bX, \bW}(S^*(U,t) , t ) \right ) \le  \frac{1}{4} g(U). 
 \end{align}
The first inequality is due to Lemma~\ref{lem:grad_inq} and the second inequality follows from \eqref{eq:MleU} and the fact that $g(U)$ is non-decreasing. Plugging this inequality back into \eqref{eq:J2_UB} completes the proof. 
\end{proof}

We are now  have all the ingredients needed for the proof of Theorem~\ref{thm:I_XW_UB}.  Starting with \eqref{eq:J_biconjugate} and then applying the bound in  Lemma~\ref{lem:J_XW_inq} allows us to write
\begin{align}
I_{\bX,\bW}(S, t) 
& =  \inf_{U \in \cU } \left\{ J_{\bX, \bW}(U , t ) + \frac{1}{2} \gtr(S U)  \right\} \notag \\
& \le  \inf_{U \in\cU } \left\{ J_{\bX}(U ) + \frac{t}{4} g(U)+ \frac{1}{2} \gtr(S U)   \right\} \label{eq:I_XW_UB_c}.
\end{align}
Note that this is a variational upper bound in terms of the dual variable $U$, which corresponds to the MMSE matrix. To rewrite this expression in terms of an infimum over the signal-to-noise matrix, we define the function $h : \psd^\ell \to \reals$ according
\begin{align*}
h(\Delta)  & \triangleq  \gtr\left(  \left(  \Gamma  R -  R^{-1} \Delta \right)^2  \right)+ \frac{1}{ n^2} \ex{  \left \| R  \bX^T \bX -   R \ex{ \bX^T \bX} \right\|_F^2}.
\end{align*}
Then, a straightforward calculation shows that $g(U)$ is the concave conjugate of $h(\Delta)$ in the following sense: 
\begin{align}
 g(U)  = \inf_{\Delta \in \psd^\ell } \left\{ h(\Delta)  +2 \gtr\left( \Delta U \right)\right\}, \label{eq:g_alt}
\end{align}
for all $0 \preceq U \preceq\Gamma$.  Plugging this characterization of $g(U)$ back into \eqref{eq:I_XW_UB_c}, and then swapping the order of the infimum with respect to $U$ and $\Delta$ leads to
\begin{align}
I_{\bX,\bW}(S, t) 
&  \le  \adjustlimits \inf_{U \in \cU }  \inf_{\Delta \in \psd^\ell } \left\{ J_{\bX}(U ) + \frac{t}{4} h(\Delta)   + \frac{1}{2} \gtr\left( (S + t \Delta) U \right)  \right\}  \notag\\
&  =  \adjustlimits   \inf_{\Delta \in \psd^\ell } \inf_{U \in \cU }\left\{ J_{\bX}(U ) + \frac{t}{4} h(\Delta)   + \frac{1}{2} \gtr\left( (S + t \Delta) U \right)  \right\}  \notag \\
&  =  \inf_{\Delta \in\mathbb{S}^\ell_{+}} \left\{ I_{\bX}(S + t \Delta ) + \frac{t}{4} h(\Delta)  \right\},  \notag 
\end{align}
where the final equality follows from \eqref{eq:J_biconjugate}. This concludes the proof of Theorem~\ref{thm:I_XW_UB}.

\section{Discussion} 
\label{sec:discussion}
The results presented in this paper recast the community detection problem as a multivariate problem making it possible to evaluate more than just traditional overall recovery tasks. By evaluating the formulas derived in Section~\ref{sec:formulas} we can now differentiate between the tasks of finding one community, all communities, and a subset of communities within a network. The formulas further allow us to identify a computational gap for regimes where certain recovery tasks should be theoretically attainable but where algorithms such as BP will fail to perform. 
\color{black}

\section*{Acknowledgment} 
The authors thank Lenka Zdeborov\'a for providing initial direction on this problem  and Jiaming Xu for helpful discussion regarding channel universality.  This was supported in part by funding from the Laboratory for Analytic Sciences (LAS) and by the NSF under Grant No.~1750362.  Any opinions, findings, conclusions, and recommendations expressed in this material are those of the authors and do not necessarily reflect the views of the sponsors.


\appendix

\section{Proofs of Results in Section~\ref{sec:formulas}}

\subsection{Proof of Theorem~\ref{thm:cF} } 

Combining Corollary~\ref{cor:GtoZ}  and Theorem~\ref{thm:matrix_factorizatiion} with $t = 1$  yields 
\begin{align*}
\lim_{n \to \infty} \frac{1}{n} I(\bX; \bG)   & = \inf_{\Delta\in \psd^\ell }  \left\{ I_{\bX}(  \Delta )  +\frac{1}{4}  \gtr\left(  \left(  \ex{ X X^T}   R -  R^{-1} \Delta \right)^2  \right) \right\},
\end{align*}
for any random matrix $\bX \in \reals^{n \times \ell}$ whose rows are drawn i.i.d.\ from a distribution  $P_X$ on $\reals^\ell$ with finite and bounded support.  Under the assumption that the rows are supported on the whitened representation described in Section~\ref{sec:whitened_basis}  it follows that $\ex{ X X^T} = I$. Furthermore, it can be verified that the infimum with respect to $\Delta$ is attained on the compact set $\{ \Delta \, : 0 \preceq \Delta \preceq R^{2}\}$ and thus the use of a minimum is justified. This concludes the proof of Theorem~\ref{thm:cF}.

\subsection{Proof of Theorem~\ref{thm:cF0} } \label{sec:thm:cF0_proof}

Combining  Corollary~\ref{cor:GtoZ} and Corollary~\ref{cor:I_XW_UB} with $t = 1$ yields 
\begin{align*}
\limsup_{n \to \infty} \frac{1}{n} I(\bX; \bG, \bY ) \le   \inf_{\Delta \in \psd^\ell} \left\{ I_{X}(S + \Delta )   +   \gtr\left(  \left(   \ex{ X X^T}  R -  R^{-1} \Delta \right)^2  \right) \right\} ,
\end{align*}
for any $S \in \psd^\ell$ and random matrix $\bX \in \reals^{n \times \ell}$ whose rows are drawn i.i.d.\ from a distribution $P_X$ on $\reals^\ell$ with finite and bounded support. Under the assumption that the rows are supported on the whitened representation described in Section~\ref{sec:whitened_basis}  it follows that $\ex{ X X^T} = I$. Furthermore, it can be verified that the infimum with respect to $\Delta$ is attained on the compact set $\{ \Delta \, : R ( I -  M_{X}(S) ) R \preceq \Delta \preceq R^2 \}$ and thus the use of a minimum is justified. This concludes the proof of Theorem~\ref{thm:cF0}.

\subsection{Proof of Theorem~\ref{thm:MMSE_UB}}
\label{sec:thm:MMSE_UB_proof}

The key idea underlying this proof is to exploit the integral form of  matrix I-MMSE relationship, which gives
\begin{align*}
I(\bX; \bG, \bY)  - I(\bX; \bG) = \frac{n}{2} \int_0^1 \gtr\left( \MMSE(\bX \mid \bG, \bY) \Big\vert_{S = S_u}  \frac{\dd}{ \dd u} S_u \right) \, \dd u,
\end{align*}
for any differentiable path $S_u$ with $S_0 = 0 $ and $S_1 = S$. Combining Theorems~\ref{thm:cF} and ~\ref{thm:cF0} provides an upper bound on the leading order terms of the left-hand side of this expression in the large-$n$ limit. We will show that this upper bound implies an asymptotic upper bound on the matrix MMSE  with respect to the Loewner partial order.

To simplify notation we let $\ell = k-1$ and define the functions $f_n:  \psd^\ell \to \reals$ and $f : \psd^n \to \reals$ according to
\begin{align*}
f_n(S) & \triangleq \frac{1}{n} I(\bX; \bG, \bY)\\
f(S) & \triangleq \min_{\Delta\in \psd^\ell}  \left\{ I_{\bX}(S + \Delta )  +\frac{1}{4} \gtr\left(  \left(  R -  R^{-1} \Delta \right)^2  \right) \right\}.
\end{align*}
For for all $S \in \psd^\ell$, the upper bound on the mutual information in Theorem~\ref{thm:cF} combined with the exact limit in Theorem~\ref{thm:cF0} allows us to write
\begin{align}
\limsup_{n \to \infty} \left\{  f_n(S) - f_n(0) \right\} \le f(S) - f(0). \label{eq:fn_limUB}
\end{align} 

The next step is to show that \eqref{eq:fn_limUB} implies an upper bound  for the gradient $\nabla f(S)$ for all positive definite $S$.  Let  $\cT = \{ T \in \psd^\ell \, : \, T \preceq I\}$.  For every $S \in \mathbb{S}_{++}^d$, $T \in \cT$ and $\eps \in (0,  \lambda_\text{min}(S)]$, we can write
\begin{align}
 \frac{1}{\eps} \left( f_n( \eps T) - f_n(0) \right) & = \frac{1}{\eps}  \int_0^\eps \gtr\left(\nabla f_n( u  T ) T \right)  \, \dd u  \ge \frac{1}{\eps} \int_0^\eps \gtr\left(\nabla f_n( S ) T \right)  \, \dd u  =  \gtr\left(\nabla f_n( S) T \right), \label{eq:fn_limUB_b}
\end{align}
where the inequality holds because $u T \preceq \eps T \preceq \eps I \preceq S$ for all $u \in [0,\eps]$ and $\nabla f_n $ is non-increasing with respect to the Loewner partial order. Meanwhile, we note that $f$ is concave because it is the poinitwise infimum of concave functions.  By the envelope theorem~\cite{milgrom:2002}, the supergradient of $f(S)$ at $S =0$ is the closure of the set $\{ \frac{1}{2}  M_{X}(\Delta) \, : \, \text{$\Delta$ attains the minimum in the definition $f$}\}$. 
Hence, 
\begin{align}
\frac{1}{\eps} \left( f(\eps T) - f(0) \right)  \le  \gtr\left( \nabla f(0) T  \right), \label{eq:fn_limUB_c}
\end{align} 
where $ \nabla f(0) $ denotes any matrix in the supergradient of $f(S)$ at $S= 0$. Combining \eqref{eq:fn_limUB},  \eqref{eq:fn_limUB_b}, and  \eqref{eq:fn_limUB_c} leads to
\begin{align}
\limsup_{n \to \infty} \gtr\left(\nabla f_n( S) T \right)  \le  \gtr\left( \nabla f(0) T  \right), \label{eq:fn_limUB_d}
\end{align}
for all $S \in \mathbb{S}_{++}^d$ and $T \in \cT$

The final step in the proof is to show that \eqref{eq:fn_limUB_d} implies an upper bound on the maximum eigenvalue of $\nabla f_n( S)  - \nabla f(0)$. To proceed,  observe that the set  $\cT$ is compact, and thus for every $\delta > 0$ there exists an integer $M$ and a set of matrices  $\{ T_1, \dots, T_M\}$ such that $\max_{T \in \cT} \min_{m \in [M]} \| T_m  - T\|_F \le \delta$. Therefore, the maximum eigenvalue can be upper bounded as follows: 
\begin{align*}
\lambda_\text{max}\left( \nabla f_n( S)  - \nabla f(0) \right) & = \max_{ T \in \cT} \gtr\left( \left( \nabla f_n( S)  - \nabla f(0) \right) T \right)  \\
 & \le  \max_{m \in [M]} \gtr\left( \left( \nabla f_n( S)  - \nabla f(0) \right) T_m \right)   + \delta \| \nabla f_n( S)  - \nabla f(0)\|_F,
\end{align*}
By \eqref{eq:fn_limUB_d}, the limit superior of the first term on the right-hand side is non-positive. Meanwhile  the gradient $\nabla f_n(S)$ is bounded uniformly with respect to $S$ and $n$. Noting that $\delta$ can be chosen arbitrarily small complete the proof of Theorem~\ref{thm:MMSE_UB}.

\subsection{Proof of Theorem~\ref{thm:MMSE_LB} } \label{proof:thm:MMSE_LB}
Given $t \in [0, \infty)$,  let $\bZ(t) = \sqrt{t/n} \bX R  \bX^T + \bm{\xi}$ where $\bm{\xi}$ is a standard Gaussian Wigner matrix.  Starting with the I-MMSE relation in \eqref{eq:IXZ_dt}, we obtain, for all $t > 0$, 
\begin{align*}
\frac{4 ( I(\bX; \bG, \bZ(t))   - I(\bX; \bG) )  }{nt}  &  = \frac{1}{t}  \int_0^t \frac{1}{n^2 } \ex{ \left\| \bX R\bX^T - \ex{ \bX R\bX^T \mid \bG , \bZ(\tau) } \right \|_F^2}  \, \dd \tau\\
&  \le  \frac{1}{n^2 } \ex{ \left\| \bX R \bX^T - \ex{ \bX R \bX^T \mid \bG } \right \|_F^2} ,
\end{align*}
where the inequality holds because the integrand is non-increasing in $\tau$. To characterize the asymptotic limit of  the left-hand side, we start with Theorem~\ref{thm:GtoZ} and use the same steps that led to Corollary~\ref{cor:GtoZ} to obtain
\begin{align}
\lim_{n \to \infty} \frac{1}{n}  \left | I(\bX; \bG, \bZ(t)  ) -  I(\bX; \bZ'(1),  \bZ(t))  \right| = 0,
\end{align}
where $\bZ'(1)$ and $\bZ(t)$ are conditionally independent given $\bX$.  By \cite[Lemma~2]{reeves:2018a}, the information provided by two independent Gaussian observations  can be expressed in terms of a  signal observation according to $I(\bX;\bZ'(1) ,  \bZ(t))  = I(\bX;\bZ'(1 + t))$. Thus we can apply Theorem~\ref{thm:matrix_factorizatiion} to obtain
\begin{align*}
\lim_{n \to \infty} \frac{1}{n} I(\bX; \bG, \bZ(t)  ) = \psi(1 + t) 
\end{align*}
where
\begin{align*}
 \psi(\gamma) \triangleq   \min_{\Delta \in \psd^\ell}  \underbrace{I_{X}(\Delta)  + \frac{1 }{ 4} \gtr\left( \left( \sqrt{ \gamma }R - \frac{1}{\sqrt{ \gamma}} R^{-1}   \Delta \right)^2 \right)}_{\cF_\gamma(\Delta)}.
\end{align*}
Putting the above pieces together, we obtain
\begin{align}
\liminf_{n \to \infty} \frac{1}{n^2 } \ex{ \left\| \bX R \bX^T - \ex{ \bX R  \bX^T \mid \bG } \right \|_F^2}  & \ge  \frac{ 4(  \psi( 1 + t )- \psi(1))}{ t},   \label{eq:matrix_MMSE_LB}
\end{align}
for all $t > 0$.

Next, we consider the limiting behavior of the right-hand side of \eqref{eq:matrix_MMSE_LB} as $t$ decreases to zero. 
Observe that the gradients of the potential function $\cF_\gamma(\Delta)$ are given by
 \begin{align}
\partial_\gamma \cF_\gamma(\Delta) & =  \frac{1}{4} \gtr (R^2  - \gamma^{-2} R^{-1} \Delta^2 R^{-1} )\\
\nabla_\Delta  \cF_\gamma(\Delta) &=  \frac{1}{2} M_{X}(\Delta)    - \frac{ 1}{2} I  + \frac{1}{2} \gamma^{-1} R^{-1} \Delta R^{-1} . \
\end{align}
Let $\cD = \arg\min \cF_{1}(\cdot)$. Starting with the envelope theorem~\cite{milgrom:2002}, we have
\begin{align}
\lim_{t \to 1^+} \frac{ 4(  \psi( 1+ t) - \psi(1))}{ t}  & =  \min_{\Delta \in \cD }  4 \partial_t  \cF_{1}( \Delta) \notag \\
&  =  \min_{\Delta \in \cD} \gtr (R^2  - R^{-1} \Delta^2 R^{-1} )\ \notag \\
& =   \min_{\Delta \in \cD} \gtr( R^2 -  R  ( I - M_{X}(\Delta))^2 R)  \label{eq:m0_LB}
\end{align}
where the last step  holds because  every $\Delta \in \cD$ is a stationary point of $\cF_{1}(\cdot)$ and thus satisfies $\Delta = R  (I- M_X(\Delta)) R $.

Combining Lemma~\ref{lem:grad_inq}, evaluated with $S =0$, with the assumption $ \frac{1}{n} \ex{ \bX \bX^T} = I$ gives
\begin{align}
 \frac{1}{n^2 } \ex{ \left\| \bX R \bX^T - \ex{ \bX R  \bX^T \mid \bG } \right \|_F^2}    \le \gtr\left( R^2 -  R  ( I - \MMSE(\bX \mid \bG))^2 R \right)   +  \ex{ \left\| \frac{1}{n}  \bX \bX^T - I \right \|^2_F} ,
  \end{align}
 where the second term on the right-hand side converges to zero in the large-$n$ limit by the law of large numbers. 

Combining  this inequality with \eqref{eq:matrix_MMSE_LB} and  \eqref{eq:m0_LB} gives
\begin{align}
\liminf_{n \to \infty}  \gtr(  R^2 -  R( I - \MMSE(\bX \mid \bG)))^2R )  \ge  \min_{\Delta \in \cD} \gtr( R^2-  R( I - M_{X}(\Delta))^2R ) . \notag
\end{align}
Rearranging the terms  completes the proof.



\section{Proof of Theorem~\ref{thm:GtoZ}}\label{proof:thm:GtoZ}

Recalling that $\bG$ is a symmetric matrix with zeros on the diagonal and entries above the diagonal drawn according to \eqref{eq:Aij}, we can write  $I(\bW; \bG) = I(\{W_{ij}\}_{i < j} ;  \{ G_{ij} \}_{i < j} )$. Meanwhile,  the fact that $\bZ$ is symmetric allows us to write
\begin{align}
I(\bW; \bZ) &= I(\{W_{ij}\}_{i < j} ;  \{ Z_{ij} \}_{i < j} ) + I(\{W_{ii} \} ;  \{ Z_{ii} \} \mid \{ Z_{ij} \}_{i < j} ), 
\label{eq:IWZdecomp}
\end{align}
where $\{W_{ii}\}$ denotes the diagonal entries of $\bW$.  By the chain rule for mutual information and the conditional independence of  $\{Z_{ij}\}_{i \le j}$ given $\bW$, the  second term on the right-hand side  of \eqref{eq:IWZdecomp} can be upper bounded as follows: 
\begin{align*}
I(\{W_{ii} \} ;  \{ Z_{ii} \} \mid \{ Z_{ij} \}_{i < j} )  \le \sum_{i=1}^n I(W_{ii};  Z_{ii}   ) \le \sum_{i=1}^n  \frac{1}{2} \log(1 + B^2/(2n))  \le B^2/4,
\end{align*} 
where the second inequality follows from the assumption $\var(W_{ij}) \le B^2/n$ and the capacity of the additive Gaussian noise channel.  In the following,  we compare $I(\bW; \bG)$ with  the first term on the right-hand side of \eqref{eq:IWZdecomp}.

To simplify notation, let $m = n(n-1)/2$ and let $W$, $G$ and  $Z$ denote the $m$-dimensional vectors obtained by stacking the columns above the diagonal in $\bW$, $\bG$, and $\bZ$, respectively. The mutual information terms of interest can then be expressed as
\begin{align*}
I(W;G) & =  \int \DKL{P_{G|W = w}}{P_{G} } \,  \dd P_{W}(w) \\
I(W;Z) & =  \int \DKL{P_{Z|W = w}}{P_{Z} } \,  \dd P_{W}(w),
\end{align*}
where $P_{G \mid W =w}$ is the conditional distribution of $G$ corresponding to a realization $w$ of $W$ and $\DKL{P}{Q}$ denotes the relative entropy between distributions $P$ and $Q$.  Our approach is to prove that the inequality 
\begin{align}
\left |  \DKL{P_{G|W = w}}{P_{G} } -  \DKL{P_{Z|W = w}}{P_{Z}}  \right|  \le  C(\delta, B)  \left(  \frac{n^{3/2}  + n  \log N }{ \sqrt{d (n-d)}}   + \frac{n\,  \log N   }{ d (n-d)}   \right), \label{eq:DPGw_inq} 
\end{align}
holds uniformly for all $w \in \reals^m$ satisfying $\|w\|_\infty \le B/ \sqrt{n}$. The desired result for the mutual information then follows from Jensen's inequality.

\subsection{Proof of Inequality~\eqref{eq:DPGw_inq} }

Condition on a realization $w $ of $W$ and let $G \sim P_{G \mid W = w}$. Let $P_U$ be the shifted distribution defined by $\dd P_U(u) = \dd P_{W}(w + u)$ and let $\cU$ denote the support of $P_U$. For each $u \in \cU$, we define the log likelihood ratio according to 
\begin{align*}
\cL(u) & \triangleq   \log \frac{ \dd P_{G \mid W }(G \mid w + u)} { \dd P_{G \mid W }(G \mid w)} = \sum_{i=1}^m  \log \frac{ \dd P_{G_i \mid W_i }(G_i \mid w_i + u_i)} { \dd P_{G_i \mid W_i }(G_i \mid w_i)}. 
\end{align*}
Using this notation,  the relative entropy be written as
\begin{align}
\DKL{P_{G|W = w}}{P_{G} }  & =  - \ex{  \log  \left( \int e^{ \cL(u) } \, \dd P_{U}(u) \right) }, \label{eq:Dint}
\end{align} 
where the expectation is with respect to $G \sim P_{G \mid W = w}$.  The score function associated with $w$ is the $m$-dimensional random vector given by $V \triangleq \nabla  \cL(0)$ and the Fisher information matrix associated with $w$  is the $m \times m$ positive semidefinite matrix  given by $\cI \triangleq \cov(V) = - \ex{ \nabla^2 \cL(0)}$. Under the Bernoulli observation model in \eqref{eq:Aij}, the entries of $V$ are independent and given by
\begin{align}
V_i 
& =  \frac{G}{ \sqrt{d/ (n-d)}  + w_i} - \frac{1-G}{ \sqrt{ (n-d)/d}  - w_i}, \label{eq:Vi}
\end{align}
and the Fisher information matrix  is diagonal with 
\begin{align}
\cI_{ii} 
& =     \frac{1  }{ (\sqrt{d/ (n-d)}  + w_i )( \sqrt{ (n-d)/d}  - w_i)}. \label{eq:cIi}
\end{align}

To proceed, we define two different approximations to the relative entropy in \eqref{eq:Dint} according to 
\begin{align*}
\widehat{D}_1  & \triangleq  - \bEx_{V}\left[   \log  \left( \int e^{ \langle u, V \rangle  - \frac{1}{2} \langle u , \cI u \rangle  } \, \dd P_{U}(u) \right)  \right] \\
\widehat{D}_2  & \triangleq -\bEx_{\tilde{V}}\left[   \log  \left( \int e^{ \langle u, \tilde{V} \rangle  - \frac{1}{2} \langle u , \cI u \rangle  } \, \dd P_{U}(u) \right) \right] 
\end{align*}
where $\tilde{V} \sim \normal(0, \cI)$ is a Gaussian random vector with the same mean and covariance as the score function $V$. By the triangle inequality, 
\begin{align*}
\left |  \DKL{P_{G|W = w}}{P_{G} } -  \DKL{P_{Z|W = w}}{P_{Z}}  \right|  \le  \left |  \DKL{P_{G|W = w}}{P_{G} } -  \widehat{D}_1 \right| + \left |  \widehat{D}_1-  \widehat{D}_2 \right| + \left | \widehat{D}_2-  \DKL{P_{Z|W = w}}{P_{Z}}  \right|. 
\end{align*}
The terms on the right-hand side are upper bounded in the following lemmas. The notation  $f(x)  = O(g(x))$ means that there is a universal constant $C$ such that $f(x) \le C g(x)$ and notation  $f(x)  = O_{B, \delta} (g(x))$ means that there is a constant $C(B, \delta)$ such that $f(x) \le C(B, \delta) \,   g(x).$

\begin{lemma} We have
\begin{align*}
\left |  \DKL{P_{G|W = w}}{P_{G} } -  \widehat{D}_1 \right|  &  = O_{B, \delta} \left( \frac{n^{3/2} +  n \sqrt{ \log N}}{ \sqrt{ d (n-d)} }  + \frac{  n \log N}{ d (n-d)}   \right) .
\end{align*}
\end{lemma}
\begin{proof}
Let $A  = (A_1, \dots, A_m)$ be  the zero-mean random vector defined by  $A_i = \partial_i^2 \cL(0)  + \cI_{ii}$ where $\partial_i^2$ denotes the second partial derivative with respect to $u_i$, and let $\{\cA(u) \, : \, u \in \cU\}$ be the random process given by  $ \cA(u)  =  \frac{1}{2} \sum_{i=1}^m u^2_i A_i.$  
The second order Tayler series expansion of $\cL(u)$ about the point $u  = 0$ can be expressed as
\begin{align*}
\cL(u) & =  \langle u, V \rangle  - \frac{1}{2} \langle u , \cI u \rangle +   \cA(u) + \cR(u), 
\end{align*}
where $\cR(u)$ is the remainder term. In view of  \eqref{eq:Dint} and the definition of $\widehat{D}_1$, it follows that 
\begin{align*}
\left |  \DKL{P_{G|W = w}}{P_{G} } -  \widehat{D}_1 \right|  & \le   \ex{ \sup_{u \in \cW} | \cA(u) |  }   +  \ex{ \sup_{u \in \cW} | \cR(u) |  } .
\end{align*}

We first consider the expected supremum of $\cR(u)$. By Taylor's theorem, there exists a vector $\tilde{u}$ between zero and $u$ such that 
\begin{align}
\cR(u) &= \frac{1}{6} \sum_{i=1}^m u_i^3 \partial_i^3 \cL(\tilde{u}). \label{eq:Remainder_decomp}
\end{align}
 Direct computation reveals that $\partial_i^3 \cL(u)   = 2  G( \sqrt{d/ (n-d)}  + w_i +u_i)^{-3} - 2(1-G)( \sqrt{ (n-d)/d}  - w_i - u_i )^{-3}$. 
Noting that  $|u_i| \le 2 B/\sqrt{n}$ for all $u \in \cU$, one obtains the uniform upper bound
\begin{align}
\ex{ \sup_{u \in \cU}  \left|  \partial_{i}^3  \cL(u)  \right| }  
 &= O_{B, \delta}\left(  \frac{n}{ \sqrt{d (n-d)}} \right) .  \label{eq:Remainder_decomp2}
\end{align}
Combining~\eqref{eq:Remainder_decomp}  and \eqref{eq:Remainder_decomp2} with the fact that  $m= O(n^2)$ and  $|u_i| \le  2 B/\sqrt{n}$ leads to   
\begin{align*}
\ex{\sup_{ u \in \cU} \cR(u) }   
= O_{B,\delta} \left(  \frac{n^{3/2} }{ \sqrt{d (n-d)}} \right). 
\end{align*}

Next, we consider the expected supremum of $\cA(u)$. Under the Bernoulli observation model in \eqref{eq:Aij}, the entries of $A$ are independent and a straightforward calculation shows that there exist numbers 
\begin{align}
\nu &=  O_{\delta, B} \left(  \frac{ n^2}{ d (n-d)}  \right)  \label{eq:nu_scale}\\
c &=  O_{\delta, B} \left( \frac{n^2}{ d (n-d)}  \right) , \label{eq:c_scale}
\end{align}
such that $\ex{ \left| A_i \right|^2} \le \nu$ and $\left| A_i \right|  \le c$ almost surely. By  Bernstein's Inequality \cite[Theorem~2.10]{boucheron:2013}, it follows that each $A_i$ is a sub-gamma random variable 
 with variance factor $\nu$ and scale factor $c$, i.e., the cumulant generating function satisfies 
\begin{align*}
\log \ex{ e^{ t A_i} } \le \frac{ \nu t^2}{ 2( 1 - ct)},
\end{align*}
for all $|t| \le c$. Hence, for all $u \in \cU$ and $|t| \le 2B ^2c/n $,
\begin{align*}
\log \ex{ e^{ t \cA(u) } } & =  \sum_{i=1}^m \log \ex{ e^{( t u_i^2 /2 ) A_i } }\\
&  \le  \sum_{i=1}^m \frac{   \nu (tu_i^2 /2)^2 }{ 2( 1 - c(t u_i^2 /2))} \\
&  \le   \frac{ 4 m B^4 n^{-2}  \nu t^2 }{ 2( 1 - 2 B^2 c n^{-1} t ))} , 
\end{align*}
where the equality follows from the independence of  the entries of $A$ and the last inequality holds because $u_i^2 \le 4 B^2/n$. An application of the maximal inequality~\cite[Corollary~2.6]{boucheron:2013} yields
\begin{align}
\ex{ \max_{u \in U}  |  \cA(u) | } &  \le \sqrt{ \frac{8 m B^4 \nu \log(2  N)}{n^2} }  + \frac{ 2 B^2 c  \log(2 N) }{n} .\label{eq:max_inq}
\end{align}
Combining \eqref{eq:max_inq} with $m =O(n^2)$ and  the scalings in \eqref{eq:nu_scale} and  \eqref{eq:c_scale} leads to the desired result. 
\end{proof}

\begin{lemma} We have
\begin{align*}
\left |  \widehat{D}_1   -  \widehat{D}_2 \right|  &  =  O_{\delta, B}  \left(  \frac{ n^{3/2}}{ \sqrt{d(n-d)} } \right).
\end{align*}
\end{lemma}
\begin{proof}
Let $\Phi : \reals^m \to \reals$ be defined as $\Phi(v)  =    -  \log \int e^{ \langle v, u \rangle  - \frac{1}{2} \langle u, \cI u \rangle } \dd P_U(u)$. Then, we can write 
\begin{align*}
 \widehat{D}_1 - \widehat{D}_2 =    \ex{ \Phi( V)} -  \bEx[ \Phi(\tilde{V} )]
\end{align*}
where we recall that $V$ has independent entries and $\tilde{V}$ is a Gaussian vector with the same first two moments as $V$. We bound this difference using the generalized Lindeberg principle~\cite[Theorem~1.1]{chatterjee:2006}, which implies that, if  there exists a constant $L$ such that $| \partial_i^3 \Phi(v) | \le L$  for each $i$ and $v$, then
\begin{align}
\left |  \ex{ \Phi(V)} -  \ex{ \Phi(\tilde{V})} \right|  \le  \frac{m L }{6}   \max_{i \in [m]}   \left(   \ex{ |V_i|^{3} }  + \ex{ |\tilde{V}_i|^{3} }   \right). \label{eq:Lindeberg}
\end{align}

From \eqref{eq:Vi}  and \eqref{eq:cIi} it can be verified that the  third moments satisfy
\begin{align}
  \ex{ |V_i|^{3} }  + \ex{ |\tilde{V}_i|^{3} }    = O_{\delta, B}  \left(  \frac{ n}{ \sqrt{d(n-d)} } \right). \label{eq:third_moments}
\end{align}
Meanwhile, if we let $A$ be a $\cU$-valued random vector drawn according to the measure 
\begin{align*}
\frac{ e^{\langle  v  , u  \rangle   -  \frac{1}{2}  \langle u , \cI  u \rangle } \dd P_U(u) }{ \int   e^{\langle  v  , \eta' \rangle   -  \frac{1}{2}  \langle u' , \cI    u' \rangle } \dd P_{U} (u') } , 
\end{align*}
then the partial derivatives of $\Phi$ can be expressed as 
\begin{align*}
\partial_i \Phi(v) &=  - \ex{ A_i } \\
\partial^2_{i}  \Phi(v) &=  -  \ex{A_i^2}  + \ex{A_i}^2  \\
\partial^3_{i}  \Phi(v) &=  - \ex{ A_i^3}  + 3 \ex{ A_i^2} \ex{A_i}  - 2\ex{ A_i}^3.
\end{align*}
Noting that  $|A_i| \le 2B/\sqrt{n}$ for all $A \in \cU$ we see that   $|\partial^3_{i}  \Phi(v) |= O_B\left(  n^{-3/2} \right)$. Combining this inequality with \eqref{eq:Lindeberg} and \eqref{eq:third_moments} completes the proof. 
\end{proof}

\begin{lemma} We have
\begin{align*}
 \left | \widehat{D}_2-  \DKL{P_{Z|W = w}}{P_{Z}}  \right| &  =  O_{\delta, B}  \left(  \frac{ n^{3/2}}{ \sqrt{d(n-d)} } \right).
\end{align*}
\end{lemma}
\begin{proof}
Let $\Psi : \psd^m \to \reals$ be defined as 
\begin{align*}
\Psi(K)  & =   -  \bEx_{N} \left[ \log \int e^{ \langle K^{1/2}   N  , u \rangle  - \frac{1}{2} \langle u, K u \rangle } \dd P_U(u) \right],
\end{align*}
where the expectation is with respect to $N \sim \normal(0, I_m)$. Then, a straightforward calculation reveals that 
\begin{align*}
\widehat{D}_2 -  \DKL{P_{Z|W = w}}{P_{Z}}  & = \Psi(\cI)  -  \Psi(I),
\end{align*}
where we recall that $\cI$ is a diagonal matrix given by \eqref{eq:cIi}.

Next, we consider the gradient of $\Psi(K)$. Let  $\mu(\cdot  \mid K, N)$ be the probability measure on $\cU$ defined by
\begin{align*}
\dd  \mu(u  \mid K, N) & =  \frac{ e^{ \langle K^{1/2} N ,  u \rangle  - \frac{1}{2} \langle u, K u \rangle }   \dd P_U(u)  }{ \int e^{ \langle K^{1/2} N ,  u' \rangle  - \frac{1}{2} \langle u', K u' \rangle } \dd P_{U}(u')}, 
\end{align*} 
and observe that
\begin{align*}
\nabla \Psi(K) 
& =  \frac{1}{2}  \ex{   \int \left(  uu^T -  K^{-1/2}N u^T     \right) \dd \mu(u \mid K, N) } . 
\end{align*}
Using Gaussian integration by parts (Stein's lemma) in conjunction with  the relation
\begin{align*}
\nabla_N  \,    \dd \mu(u  \mid K, N) & =  \left(  K^{1/2} u - \int  u' \,  \dd \mu(u' \mid K, N) \right)  \,    \dd \mu(u \mid K, N),
\end{align*} 
leads to 
\begin{align*}
\nabla \Psi(K) & =  \frac{1}{2}  \ex{   \int  u\,  \dd \mu(u \mid K, N) \left(  \int   u\,  \dd \mu(u \mid K, N) \right)^T  }. 
\end{align*}
This identity implies that the nuclear norm of the gradient is bounded by
\begin{align*}
\| \nabla \Psi(K)\|_\star  = \gtr\left(  \nabla \Psi(K) \right)   = \frac{1}{2}  \ex{\left \| \int  u\,  \dd \mu(u \mid K, N) \right \|^2} \le    \sup_{u \in \cU} \frac{1}{2}  \|u \|^2 \le \frac{2 m B^2}{ n} 
\end{align*}
where the last step holds because $\|u\| \le \sqrt{m} 2 B/\sqrt{n}$ for all $u \in \cU$. 

With these results in hand, we can now write
\begin{align*}
\left | \Psi(\cI)  -  \Psi(I)\right| & = \left |  \int_0^1  \frac{\dd}{ \dd t} \Psi( t \cI - (1- t) I)  \, \dd t \right| \\
&  = \left |  \int_0^1  \gtr\left( \nabla \Psi( t \cI - (1- t) I)  ( \cI - I )  \right)  \dd t \right| \\
&  \le \left |  \int_0^1  \|  \nabla \Psi( t \cI - (1- t) I) \|_\star  \|\cI - I \|  \, \dd t\right| \\
&\le \frac{2 m B^2}{ n}  \| \cI - I\|.
\end{align*}
Finally, from \eqref{eq:cIi}, it can be verified that
\[
\|\cI - I\|  = O_{B,\delta} \left( \frac{\sqrt{n}}{ \sqrt{ d (n-d)}}  \right), 
\]
which completes the proof. 
\end{proof}

\color{black}

\section{Derivation of Theorem~\ref{thm:matrix_factorizatiion}} \label{sec:thm:matrix_factorizatiion_proof}
First we observe that if $R$ is positive definite then $R^{1/2}$ is well defined. Introducing the transformed representation $\tilde{\bX} =  \bX R^{1/2}$, we can then write 
\begin{align}
\bW  = \frac{1}{\sqrt{n}} \bX R  \bX^T  = \frac{1}{\sqrt{n}} \tilde{\bX} \tilde{\bX}^T \label{eq:Wtilde}.
\end{align}
Note that if $R$ is negative definite then the same decomposition holds with $(-R)^{1/2}$. This transformation shows that it is sufficient to focus on setting where $R$ is the identity matrix.

The result given in \cite[Theorem~12]{lelarge:2018} is stated as follows: 
\begin{align*}
\lim_{n \to \infty} \frac{1}{n} I\left(\bX ;  \sqrt{\frac{t}{n}}\bX \bX^T +\bm{\xi} \right )  = \inf_{S \in \mathbb{S}_+^{d} } \tilde{\cF}_t(S) ,
\end{align*}
where
\begin{align*}
\tilde{\cF}_t(S) & =  \frac{t}{4} \left\| \bEx \left[ X X^T \right] \right \|_F^2  + \frac{t}{4} \|S\|_F^2 -  \ex{ \log\left(  \int  \dd P_{X}(x)  \exp\left(   \sqrt{t} N^T S^{1/2} x  + t X^T S x -   \frac{t}{2} x^T S x \right)  \right)},
\end{align*}
with $N \sim \normal(0, I_d)$  independent of $X \sim P_{X}$. To see that this expression is equivalent to the on given in Theorem~\ref{thm:matrix_factorizatiion}, observe that the mutual information function $I_{X}(S)$ can be expressed as follows: 
\begin{align*}
I_{X}(S) 
& = \ex{ \log\left( \frac{\exp\left( - \frac{1}{2} \| N\|_F^2    \right)  }{ \int  \dd P_{\tilde{X}}(x)  \exp\left( - \frac{1}{2} \| N+ S^{1/2} X  - S^{1/2} x\|_F^2    \right) }  \right) }\\
& = -  \ex{ \log\left(  \int  \dd P_{X}(x)  \exp\left(   N^T S^{1/2} x  + X^T S x -  \frac{1}{2} x^T S x \right)  \right)} + \frac{1}{2} \gtr\left(S \ex{ X X^T}\right).
\end{align*}
Rearranging  terms leads to
\begin{align*}
\tilde{\cF}_t(S) 
 & = I_{X}(t S)  + \frac{t}{ 4} \gtr\left( \left(\ex{X X^T  } - S \right)^2 \right). 
 \end{align*}
Finally, using the scaling relationship  $I_{R^{1/2} X}(S)= I_{X}(R^{1/2}SR^{1/2} )$ leads to the version of the result stated in Theorem~\ref{thm:matrix_factorizatiion}.


\section{Mutual Information and MMSE in Gaussian Noise}\label{sec:I-MMSE}

\subsection{Linear Gaussian Channel}\label{sec:gaussian_channel}

The scalar I-MMSE relationship~\cite{guo:2005a} asserts the the derivative of mutual information in a Gaussian noise channel with respect to the inverse noise variance is equal to one half times the MMSE. A recent line of work in the information theory literature has focused on  multivariate extensions of this result for linear Gaussian channel~\cite{guo:2005a, palomar:2006,lamarca:2009, payaro:2009}. This section  briefly reviews some of  results described by the first author and others~\cite{reeves:2018a}. Given a random vector $X \in \reals^d$ the functions $I_{X} : \psd^d \to [0, \infty)$ and $M_{X} : \psd^d \to  \psd^d$ are defined as \cite{reeves:2018a}:
\begin{align*}
I_X(S) & = I(X ; Y), \\
M_{X}(S) & = \ex{ \cov(X \mid Y)},
\end{align*}
where $Y = S^{1/2} X + N$ with independent Gaussian noise $N \sim \normal(0, I_d)$. These functions have a number of important properties. The function $I_X(S)$ is concave~\cite[Theorem~1]{reeves:2018a} and the matrix version of I-MMSE relation is given by $\nabla I_{X}(S) = \frac{1}{2} M_{X}(S)$ \cite[Lemma~4]{reeves:2018a}. Furthermore, these functions are able to characterize a linear Gaussian channel characterized by an arbitrary matrix $A \in\reals^{m \times n}$ via the following relationship \cite[Lemma~1]{reeves:2018a}:
\begin{align}
I(X ; A X + N') &= I_{X}(A^T A) \label{eq:I_A},
\end{align}
where $N' \sim \normal(0, I_m)$ is independent of $X$.

\subsection{Linear Gaussian Channel with Matrix Input}\label{sec:matrix_inputs}

The properties of the mutual information and MMSE described in Section~\ref{sec:gaussian_channel} extend naturally to the setting where the input is an $n \times d$ random matrix $\bX = [X_1, \dots, X_n]^T$ and the observations are given by $\bY = \bX S^{1/2} + \bN$ where $S \in \psd^d$ and $\bN$ is an $n \times d$ standard Gaussian matrix. In this setting, we define the functions:
\begin{align*}
I_{\bX}(S) & = I(\bX ; \bY), \\
M_{\bX}(S) & = \sum_{i=1}^n  \ex{ \cov(X_i \mid \bY)}. 
\end{align*} 
Using vectorization, the mutual information function can be expressed equivalently as
\begin{align}
I_{\bX}(S) & =I_{\gvec(\bX)}(I_n \otimes S), \label{eq:IX_vec}
\end{align}
where $\gvec(\bX)$ denotes the $nd \times 1$ vector obtained by stacking the columns in $\bX$ and $\otimes$ denotes the Kronecker product and. From this relationship, one finds that the I-MMSE relation still holds for matrix inputs, that is  $
\nabla I_{\bX}(S) = \frac{1}{2} M_{\bX}(S)$.

Next, we consider a useful representation of the MMSE matrix $M_{\bX}(S)$. Let $\bA$ and $\bB$ denote conditionally  independent draws form the posterior distribution of $\bX$ given $\bY$. Then, the conditional covariance can be expressed as
\begin{align*}
\cov(X_i \mid \bY) = \ex{ X_i X_i^T \mid \bY } - \ex{ A_i B_i^T} 
\end{align*}
and taking the expectation with resect to $\bY$ gives
\begin{align*}
 \ex{ \cov(X_i \mid \bY)} = \ex{ X_i X_i^T} - \ex{ A_i B_i^T}.
\end{align*}
Summing over the indices leads to
\begin{align}
M_{\bX}(S) & =\ex{ \bX^T \bX} - \ex{ \bA^T \bB} . \label{eq:MX_AB}
\end{align}

\subsection{Symmetric Matrix Estimation} \label{sec:I_MMSE_sym}

In the symmetric matrix estimation problem,  the goal it estimate an unknown matrix $\bX \in\reals^{n \times d}$ from observations of the form 
\begin{align}
\bZ  = \bX R \bX^T + \bm{\xi}, \label{eq:sym_matrix_model}
\end{align}
where $R \in \mathbb{S}^d$ is known and  $\bm{\xi} \in \mathbb{S}^n$ is a standard Gaussian Wigner matrix.  In this section, we show that this model can be viewed as special case of the linear Gaussian channel associated with matrix input given by the tensor product $\bX  \otimes \bX$, and thus the mutual information and MMSE can be characterized using the functions introduced in Sectioin~\ref{sec:matrix_inputs}

The first step is to observe that symmetric noise model given in  \eqref{eq:sym_matrix_model} provides the same information as the following asymmetric noise model:
\begin{align}
\tilde{\bZ} =\frac{1}{\sqrt{2}}  \bX R \bX^T + \bN, \label{eq:sym_matrix_model_alt}
\end{align}
where $\bN$ is an $n \times n$ standard Gaussian matrix. To see why,  note that  $\tilde{\bZ}$ can be decomposed uniquely in terms of the symmetric matrix   $(\tilde{\bZ} + \tilde{\bZ}^T)/\sqrt{2} = \bX R \bX^T  + (\bN + \bN^T)/\sqrt{2} $ and the antisymmetric matrix $(\tilde{\bZ} - \tilde{\bZ}^T )/\sqrt{2} =(\bN - \bN^T)/2$. By the orthogonal invariance of the Gaussian distribution, the antisymmetric matrix is independent of both  $\bX$ and $(\tilde{\bZ} + \tilde{\bZ}^T)/\sqrt{2}$. Noticing that $(\bN + \bN^T)/\sqrt{2}$ is a standard Gaussian Wigner matrix shows that $I(\bX  ; \bZ) = I(\bX  ; \tilde{\bZ})$.

The next step is to use vectorization to represent  the observation model  in \eqref{eq:sym_matrix_model_alt} as a linear Gaussian channel with matrix input:
\begin{align*}
\gvec(\tilde{\bZ})  = \frac{1}{\sqrt{2}} (\bX \otimes  \bX) \gvec(R)   + \gvec(\bN).
\end{align*}
In view of both \eqref{eq:I_A} and \eqref{eq:IX_vec}, the mutual information can be expressed as
\begin{align*}
I(\bX  ; \bZ) & =  I(\bX \otimes \bX  ; \bZ)  = I_{\bX \otimes \bX }\left( \frac{1}{2} \gvec(R) \gvec(R)^T \right),
\end{align*}
where the first equality holds because $\bX \otimes \bX$ is a deterministic function of $\bX$. 

This characterization of the mutual information is useful because it allows us to compute gradients with respect to the matrix $R$.  By the I-MMSE relation and the chain rule, 
\begin{align}
\nabla_{\gvec(R)}  I(\bX  ; \bZ)  = \frac{1}{2}  M_{\bX \otimes \bX}\left( \frac{1}{2} \gvec(R) \gvec(R)^T \right) \gvec(R). \label{eq:IXZ_grad}
\end{align}
Furthermore, by \eqref{eq:MX_AB}, the MMSE matrix can be expressed as
\begin{align*}
M_{\bX \otimes \bX}\left( \frac{1}{2} \gvec(R) \gvec(R)^T \right)& = \ex{ (\bX \otimes  \bX)  (\bX \otimes \bX)^T}  - \ex{ (\bA \otimes  \bA) (\bB \otimes  \bB)^T} \\
& = \ex{ (\bX \bX^T) \otimes  (\bX \bX^T)}  - \ex{ (\bA \bB^T)  \otimes (\bB \bA^T)},
\end{align*}
where $\bA$ and $\bB$ denote conditionally independent draws from the posterior distribution of $\bX$ given $\bZ$.  Therefore, \eqref{eq:IXZ_grad} can be rewritten compactly as 
\begin{align*}
\nabla_{R}  I(\bX  ; \bZ) 
& =\frac{1}{2}  \left( \ex{ \bX^T \bX  R \bX^T \bX} -   \ex{ (\bA^T \bB) R  (\bB^T \bA)} \right).
\end{align*}
Finally, if we consider the parameterization  $R_t = \sqrt{t} R$ for some $t \ge 0$, then the partial derivative with respect to $t$ is given by
\begin{align}
\frac{\partial} {\partial t}  I(\bX  ; \bZ) 
& = \frac{1}{4} \gvec(R)^T M_{\bX \otimes \bX}\left( \frac{1}{2} \gvec(R) \gvec(R)^T \right) \gvec(R) \notag\\
& =\frac{1}{4} \gtr \left( \ex{ R \bX^T \bX  R \bX^T \bX} -   \ex{R  (\bA^T \bB) R  (\bB^T \bA)} \right). \label{eq:IXZ_dt}
\end{align}

\bibliographystyle{IEEEtran}

\bibliography{arxiv_full.bbl}

\end{document}